\documentclass[11pt]{article}
\usepackage[margin=1in]{geometry}
\usepackage{amsthm}
\usepackage{amsmath}
\usepackage{amssymb}
\usepackage{enumerate}
\usepackage{tikz}
\usepackage[T1]{fontenc}
\usepackage[utf8]{inputenc}
\usepackage[ruled,vlined]{algorithm2e}
\usepackage[hidelinks]{hyperref}
\usepackage{cite}
\usepackage{authblk}
\newcommand{\C}{\mathbf{C}}
\newcommand{\N}{\mathbb{N}}
\newcommand{\Q}{\mathbb{Q}}

\newcommand{\Z}{\mathbb{Z}}
\DeclareMathOperator{\Dim}{Dim}
\newcommand{\dimFS}[1]{\dim_{\textup{FS}}^{#1}}
\newcommand{\DimFS}[1]{\Dim_{\textup{FS}}^{#1}}
\newcommand{\mh}{\textup{MH}}
\newcommand{\strong}{\textup{str}}

\newtheorem{theorem}{Theorem}
\numberwithin{theorem}{section}

\newtheorem{corollary}[theorem]{Corollary}
\newtheorem{lemma}[theorem]{Lemma}
\newtheorem{observation}[theorem]{Observation}

\theoremstyle{remark}
\newtheorem{claim}[theorem]{Claim}

\theoremstyle{definition}
\newtheorem{definition}[theorem]{Definition}

\newenvironment{claimproof}{\emph{Proof.}}{\hfill$\triangleleft$}

\title{Multi-Head Finite-State Dimension}
\author{Xiang Huang\footnote{This research was supported in part by U.S. Department of Energy Grant DE-SC0024278.}}
\affil{University of Illinois Springfield}

\author{Xiaoyuan Li\footnote{This research was supported in part by National Science Foundation Grants 1545028 and 1900716.}}

\author{Jack H. Lutz\footnote{This research was supported in part by National Science Foundation Grants 1545028 and 1900716.}}
\affil{Iowa State University}

\author{Neil Lutz}
\affil{Swarthmore College}

\date{\vspace{-3em}}

\begin{document}
\maketitle

\begin{abstract}
We introduce multi-head finite-state dimension, a generalization of finite-state dimension in which a group of finite-state agents (the heads) with oblivious, one-way movement rules, each reporting only one symbol at a time, enable their leader to bet on subsequent symbols in an infinite data stream. In aggregate, such a scheme constitutes an $h$-head finite state gambler whose maximum achievable growth rate of capital in this task, quantified using betting strategies called gales, determines the multi-head finite-state dimension of the sequence.  The 1-head case is equivalent to finite-state dimension as defined by Dai, Lathrop, Lutz and Mayordomo (2004). In our main theorem, we prove a strict hierarchy as the number of heads increases, giving an explicit sequence family that separates, for each positive integer $h$, the earning power of $h$-head finite-state gamblers from that of $(h+1)$-head finite-state gamblers. We prove that multi-head finite-state dimension is stable under finite unions but that the corresponding quantity for any fixed number $h>1$ of heads---the $h$-head finite-state predimension---lacks this stability property.

\end{abstract}

\section{Introduction}\label{sec: introduction}

    This paper introduces \emph{multi-head finite-state dimension} and \emph{multi-head finite-state strong dimension}, two new tools for algorithmic information theory. Our main theorem is a hierarchy result establishing that each additional head strengthens these multi-head dimensions.
    
    The primary historical steps leading to our work are as follows.
    \begin{enumerate}[(i)]
        \item Hausdorff~\cite{Haus19} introduced \emph{Hausdorff dimension}, which became a central tool in classical geometric measure theory. (Note the term ``classical'' here does not mean ``old,'' but rather not involving the theory of computation or related aspects of mathematical logic.)
        \item Tricot~\cite{Tric82} and Sullivan~\cite{Sull84} introduced \emph{packing dimension}, a kind of dual and upper bound of Hausdorff dimension.
        \item J. Lutz~\cite{DCC,DISS} proved a \emph{gale characterization} of Hausdorff dimension that enables \emph{effectivizations} of Hausdorff dimension at various levels of computability and computational complexity.
        \item Dai, Lathrop, J. Lutz, and Mayordomo~\cite{DDLM04} showed that Hausdorff dimension can even be effectivized at the finite-state level.
        \item Athreya, Hitchcock, J. Lutz, and Mayordomo~\cite{AHLM07} proved that the above effectivizations can also be carried out for packing dimension, thus yielding \emph{strong dimensions} at each of these levels.
    \end{enumerate}
    All the above have yielded many applications in algorithmic information theory, computational complexity, geometric measure theory, analytic number theory, and other areas. In particular, the finite-state dimension $\dimFS{}(S)$ and the finite-state strong dimension $\DimFS{}(S)$ of infinite sequences $S$ over a finite alphabet $\Sigma$ have had useful interactions with Borel normal numbers and uniform distribution theory~\cite{Buge2012,ClaRay2025,DoLuNa2007,GLLM2014,KozShe2021,LutMay2021,Lope2006,NanPul2025,NanPul2024,NanVan2016,NanVis2020}, Shannon entropy and Kullback-Leibler divergence~\cite{BoHiVi2005,HitVin2006,HLMS2021,KozShe2019,KozShe2021,DFRD}, and efficient prediction of data streams~\cite{Druc2013,ForLut2005,Hitc2003}. The references here are representative but far from exhaustive. The \href{https://www.eecs.uwyo.edu/~jhitchco/bib/dim/}{Effective Fractal Dimension Bibliography} maintained by Hitchcock is more comprehensive.

    The present paper was motivated by a long-standing open question and inspired by an ingenious proof. The long-standing open question is whether the finite-state setting is the only level of effectivity in which randomness is equivalent to full dimension. Bourke, Hitchcock, and Vinodchandran~\cite{BoHiVi2005} noted that the equivalence of finite-state randomness and finite-state full dimension follows from a theorem of Schnorr and Stimm~\cite{SchSti72}. In contrast, there is no such equivalence for algorithmic randomness and dimension, computable randomness and dimension, or resource-bounded randomness and dimension; in all these settings, there exist non-random sequences with full dimension.

    The search for other settings where this equivalence holds has focused on notions of randomness and dimension that are ``highly effective,'' defined using gamblers that are computationally weaker than polynomial-time algorithms. Examples include the settings of pushdown gamblers, Lempel-Ziv gamblers, and perhaps logspace gamblers~\cite{AlMaMo2017,AMMP2008,DotNic2007,Lope2006,Nich2004,MaMoPe2011}, but it is not yet known whether randomness and full dimension are equivalent in those settings. Here our objective is to investigate a setting that is as close as possible to, but not quite as effective as, the setting of finite-state gamblers.

    The ingenious proof that inspired our choice of a setting to investigate is due to Becher, Carton, and Heiber~\cite{BCH2018}. Translated into the original gambling (i.e., gale) formulation of finite-state dimension~\cite{DDLM04}, Becher, Carton, and Heiber proved that there is an infinite binary sequence $S$ with the following two properties.
    \begin{enumerate}[(1)]
        \item No finite-state gambler $G$ can win unbounded money betting on the successive bits of $S$.
        \item There is a gambler $\widehat{G}$ that wins money exponentially quickly betting on the successive bits of $S$. This gambler $\widehat{G}$ is ``almost'' finite-state, except that it is equipped with an additional \emph{probe to the past}, so that at any time, $\widehat{G}$ knows not only the last bit of the string scanned so far, but also the ``middle'' bit of this string (adjusted appropriately for the fact that this string may have even or odd length).
    \end{enumerate}
            
    Taken together, items (1) and (2) above imply that equipping a finite-state gambler with a ``probe to the past'' as above can make it more powerful than any finite-state gambler. This is intuitively unsurprising, because implementing such a probe requires access to an unbounded amount of past data.

    In a line of research over sixty years old~\cite{RabSco1959,Rose1966}, a device like $\widehat{G}$ is a \emph{two-head finite automaton} whose first head is the ``leading'' (scanning) head of $\widehat{G}$ and whose second head is what we called above a ``probe to the past.'' The \emph{multi-head finite automata} of this literature (surveyed nicely by Holzer, Kutrib, and Malcher~\cite{HKM09}) are typically used to recognize sets of finite strings rather than to compress or gamble on infinite data streams, but the present paper should be regarded as a continuation of this literature as well as a further development of algorithmic dimensions.

    As noted above, the goal here is to investigate a setting that is as close as possible to, but not quite as effective as, the setting of finite-state gamblers.  Accordingly, we require the trailing heads in our multi-head finite-state gamblers to be as weak as possible: They can only move forward (or hold still), and they are \emph{oblivious} in the sense that their movements do not depend on the data that they are scanning. (We are using the terminology of oblivious Turing machines here~\cite{PipFis79}. In the terminology of~\cite{HKM09}, our trailing heads are \emph{data-independent}.) We trust that other, less restrictive variants of our multi-head model will be investigated in due time.

    Section~\ref{sec:mfsg} below defines, for each positive integer $h$, the notion of an \emph{$h$-head finite-state gambler} (\emph{$h$-FSG}) in a straightforward way such that a 1-FSG is a finite-state gambler as in~\cite{DDLM04} and the gambler $\widehat{G}$ above is a 2-FSG. Section~\ref{sec:mfsd} then uses section~\ref{sec:mfsg} to define, for each positive integer $h$ and each sequence $S\in\Sigma^\omega$, the \emph{$h$-head finite-state predimension} $\dimFS{(h)}(S)$ and the \emph{$h$-head finite-state strong predimension} $\DimFS{(h)}(S)$ in such a way that $\dimFS{(1)}(S)=\dimFS{}(S)$ is the finite-state dimension of $S$ as in~\cite{DDLM04} and $\DimFS{(1)}(S)=\DimFS{}(S)$ is the finite-state strong dimension of $S$ as in~\cite{AHLM07}. It is immediately evident that each of these predimensions is, for fixed $S$, a nonincreasing function of $h$, leading immediately to the definitions of the \emph{multi-head finite-state dimension}
    \[\dimFS{\mh}(S)=\inf_{h\in\Z^+}\dimFS{(h)}(S)\]
    and the \emph{multi-head finite-state strong dimension}
    \[\DimFS{\mh}(S)=\inf_{h\in\Z^+}\DimFS{(h)}(S).\]

    We note that our model is definitionally robust, in that N. Lutz~\cite{Lutz26} has very recently shown that our multi-head finite-state predimensions and dimensions are equivalent to natural notions of multi-head finite-state \emph{compressibility}.

    We prove the main theorem of this paper, a hierarchy theorem, in section~\ref{sec:hierarchy}. Recalling that, for each $h\in\Z^+$ and each $S\in\Sigma^\omega$,
    \[\dimFS{(h+1)}(S)\leq\dimFS{(h)}(S)\quad\text{and}\quad \DimFS{(h+1)}(S)\leq\DimFS{(h)}(S),\]
    our main theorem says that each of these inequalities is strict for some sequences $S$. 

    This hierarchy theorem is analogous to Yao and Rivest's nontrivial 1978 proof that, in the context of language recognition, ``$k+1$ heads are better than $k$''~\cite{YR78} and the subsequent line of hierarchy theorems for the recognition power of multi-head finite automata in two-way~\cite{monien1980two, holzer2002multi}  or data-independent~\cite{holzer2002multi,holzer1998data,duris2012note,vdurivs2020tight} settings.
    
    For each positive integer $h$, we define a function $F_{h+1}:\Sigma^\omega\to\Sigma^\omega$, where every sequence $F_{h+1}(S)$ has a structural pattern that is easily exploited by an $(h+1)$-FSG, but each of those $h+1$ heads is essential for this task. We prove that if $R$ is a Martin-L\"of random sequence, then no $h$-FSG can gain any significant advantage from the patterns in $F_{h+1}(R)$. Our proof relies on partitioning $F_{h+1}(R)$ into a family of strings of exponentially increasing length, where each string has high conditional Kolmogorov complexity given the information that can be accessed by any $(h-1)$ heads while the leading head reads that string.

    Each fractal dimension $\texttt{dim}$ (we use this font generically in this paragraph) assigns a dimension $\texttt{dim}(E)$ to each \emph{subset} $E$ of some metric space, which is $\Sigma^\omega$ in this paper. If $\texttt{dim}$ is one of the effective dimensions of algorithmic information theory, one then sets $\texttt{dim}(S)=\texttt{dim}(\{S\})$ for each $S\in\Sigma^\omega$, as we do here. (This would be senseless for \emph{classical} fractal dimensions, because there all singletons have dimension 0.) In general, each fractal dimension $\texttt{dim}$
    has the \emph{stability} property that, for all sets $E$ and $F$,
    \[\texttt{dim}(E\cup F)=\max\{\texttt{dim}(E),\texttt{dim}(F)\}.\]
    This distinguishes fractal dimensions, for example, from measures.

    The finite-state dimensions $\dimFS{}$ and $\DimFS{}$ are stable in the above sense. However, we prove in section~\ref{sec:instability} that, for each $h\geq 2$, the $h$-head finite-state predimensions $\dimFS{(h)}$ and $\DimFS{(h)}$ are \emph{not} stable. We prove in section~\ref{sec:stability} that the multi-head finite-state dimensions $\dimFS{\mh}$ and $\DimFS{\mh}$ \emph{are} stable. This is why we use the ``predimension'' terminology.

\section{Preliminaries}\label{sec:preliminaries}

\paragraph{String and Sequence Notation}
Given a finite alphabet $\Sigma$ of symbols, the sequence space $\Sigma^\omega$ is the set of all infinite sequences over $\Sigma$; when $\Sigma=\{0,1\}$, this is the Cantor space $\C$. For $i\in\N$, we write $S[i]$ for the symbol at position $i$ in a sequence $S\in\Sigma^\omega$; the leftmost symbol of $S$ is $S[0]$. For $i,j\in\N$ with $i\leq j$, we write $S[i..j]$ for the string $S[i]\cdots S[j]\in\Sigma^{j-i+1}$.

\paragraph{\texorpdfstring{$s$}{s}-Gales and Martingales}
Following the standard template for effectivizations of fractal dimensions~\cite{DCC,DISS,AHLM07}, we will base our dimension notions on the success of functions called gales.
\begin{definition}
    For $s \in [0, \infty)$, an \emph{$s$-gale} is a function $d:\Sigma^* \to [0,\infty)$ that satisfies the condition
        \[d(w) = |\Sigma|^{-s}\sum_{b \in \Sigma}d(wb)\]
        for all $w \in \Sigma^*$. A \emph{martingale} is a $1$-gale.
\end{definition}
Intuitively, the value of a gale is interpreted as the capital of a gambler betting on successive symbols in a sequence. The number $s$ parameterizes the favorability of the betting environment. A martingale, where $d(w)$ is the average, over all symbols $b$, of $d(wb)$, corresponds to a perfectly fair betting environment, in which a maximally conservative gambler who places equal bets on all possible one-symbol extensions of $w$ can maintain constant capital. Parameters $s<1$ correspond to less favorable betting environments where ``the house takes a cut.'' In these environments, the gambler would need to exploit some knowledge about the sequence---or be lucky---to maintain constant capital. For a gale to \emph{succeed}, though, we require that the gambler's capital is unbounded, not merely constant.
\begin{definition}
    Let $d$ be an $s$-gale.
    \begin{enumerate}
        \item We say that $d$ \emph{succeeds} on a sequence $S \in \Sigma^{\omega}$ if 
        $\limsup_{n \rightarrow \infty} d(S[0..n-1]) = \infty$.
        The \emph{success set} of $d$ is $S^{\infty}[d]=\{S \in \Sigma^{\omega} \mid d \text{ succeeds on } S\}$.
        
        \item We say that $d$ \emph{succeeds strongly} on a sequence $S \in \Sigma^{\omega}$ if $\liminf_{n \rightarrow \infty} d(S[0..n-1]) = \infty$.
        The \emph{strong success set} of $d$ is $S_{\strong}^{\infty}[d]=\{S \in \Sigma^{\omega} \mid d \text{ succeeds strongly on } S\}$.
    \end{enumerate}
\end{definition}
In general, succeeding on a sequence becomes more difficult as the sequence becomes more complex, as the computational resources available to the gale become more restricted, and as the value $s$ decreases.

\paragraph{Kolmogorov Complexity}
While the success of gales quantifies the predictability of sequences, Kolmogorov complexity quantifies the compressibility of strings. These are closely related concepts, and several of our proofs rely on bounding the compressibility of parts of a sequence to derive bounds on the sequence's overall predictability. Our proofs will use the basic properties of Kolmogorov complexity stated below. Further details and background on Kolmogorov complexity can be found in~\cite{LiVit19}.

Let $U$ be a fixed, universal prefix Turing machine. For strings $\sigma,\tau\in\{0,1\}^*$, the (prefix) \emph{conditional Kolmogorov complexity} of $\sigma$ given $\tau$ is
\[K(\sigma\mid \tau)=\min\{|\pi|\mid \pi\in\{0,1\}^*\text{ and }U(\pi,\tau)=\sigma\},\]
that is, the minimum length of a \emph{program} $\pi$ for $\sigma$ given $\tau$ as an auxiliary input. The \emph{Kolmogorov complexity} of $\sigma$ is $K(\sigma)=K(\sigma\mid\lambda)$,
where $\lambda$ denotes the empty string.
\begin{itemize}
\item \emph{Computable functions do not add information:} For every computable function $f$ and all $\sigma,\tau\in\{0,1\}^*$,
\begin{equation}\label{eq:dpi}
	K(f(\sigma)\mid \tau)\leq K(\sigma\mid\tau)+O(1).
\end{equation}
\item \emph{Symmetry of information:} For all $\sigma,\tau\in\{0,1\}^*$,
\begin{equation}\label{eq:soi}
	K(\sigma,\tau)=K(\sigma\mid \tau,K(\tau))+K(\tau)+O(1).
\end{equation}
\end{itemize}
A sequence $R\in\C$ is \emph{Martin-L\"of random} if there is some constant $c\in\N$ such that, for all $n\in\N$, we have $K(R[0..n-1])\geq n-c$; in this sense, such sequences are \emph{incompressible}. We will construct sequences based on Martin-L\"of random sequences, and our constructed sequences will inherit some of this incompressibility.

\section{Multi-Head Finite-State Gamblers}\label{sec:mfsg}

We now define multi-head finite-state gamblers. These are a generalization of the 1-account finite-state gamblers defined in~\cite{DDLM04}; our definition here is self-contained. We first define the formal components of a gambler, then describe its movements and operation.

The gambler's state space will be the Cartesian product of two sets, $T$ and $Q$. The $T$ component of the state will govern the heads' movements, while the $Q$ component will govern its betting behavior. We will place a requirement on the transition function $\delta$ that it can be decomposed into separate transition functions $\delta_T$ and $\delta_Q$ that act on $T$ and $Q$, respectively, where $\delta_T$ does not depend on the symbols read. This has the effect of isolating the $T$ component of the gambler's state from the tape contents, which makes the head movements data-independent.
\begin{definition}\label{mhg}
Let $h \in \Z^+$. An \emph{$h$-head finite-state gambler} (\emph{$h$-FSG}) is a 7-tuple
\[G = (T \times Q, \Sigma, \delta, \mu, \beta, (t_0, q_0), c_0),\]
where
\begin{itemize}
    \item $T$ and $Q$ are nonempty finite sets, and $T\times Q$ is the set of \emph{states};
    \item $\Sigma$ is the finite alphabet, satisfying $|\Sigma| \geq 2$;
    \item $\delta: T \times Q \times \Sigma^h \rightarrow T \times Q$ is the \emph{transition function}, with the following property: there are functions $\delta_T:T\to T$ and $\delta_Q:Q\times\Sigma^h\to Q$ such that, for all $t\in T$, $q\in Q$, and $\vec{\sigma}\in\Sigma^h$,
    \[\delta(t,q,\vec{\sigma})= (\delta_T(t), \delta_Q(q,\vec{\sigma}));\]
    \item $\mu: T \rightarrow \{0,1\}^{h-1}$ is the \emph{movement function};
    \item $\beta: Q \rightarrow \Delta_\mathbb{Q}(\Sigma)$ is the \emph{betting function}, where $\Delta_\Q(\Sigma)$ denotes all rational-valued discrete probability distributions on $\Sigma$;
    \item $(t_0, q_0) \in T \times Q$ is the \emph{initial state}; and
    \item $c_0$ is the \emph{initial capital}.
\end{itemize}
\end{definition}

\paragraph*{Head Movements and State Transitions}
Intuitively, we regard $G$ as having $h$ \emph{heads}: a \emph{leading head} which advances through the input sequence in discrete steps, and $h-1$ \emph{trailing heads}, some subset of which advance in each step. The positions of the heads are governed by the \emph{positional states}---the members of $T$---and the function $\mu$, which dictates, in each positional state, which of the trailing heads advance. More formally, for $r\in\N$, let $\delta_T^{r}:T\to T$ denote $r$ iterated applications of $\delta_T$. Then the \emph{positional vector} of $G$ after $n$ steps, which tells us the positions of all trailing heads, is determined by the function $\pi:\N\to\N^{h-1}$ defined by $\pi(0)=(0,\ldots,0)$ and, for all $n\in\N$, recursively by $\pi(n+1)=\pi(n)+\mu\left(\delta_T^{n}(t_0)\right)$.

For $n\in\N$ and $S\in\Sigma^\omega$, let $(\pi_1,\ldots,\pi_{h-1})=\pi(n)$ and $\vec\sigma_n=(S[\pi_1],\ldots,S[\pi_{h-1}],S[n])$.
Then the state sequence $(t_0,q_0),(t_1,q_1),(t_2,q_2),\ldots$ of $G$ on input sequence $S\in\Sigma^\omega$ is given, for all $n\in\N$, by
\begin{equation}\label{eq:trajectory}
    (t_{n+1},q_{n+1})=\delta(t_n,q_n,\vec\sigma_n).
\end{equation}

Note that the head movements are \emph{oblivious} in the sense of being data-independent, which, together with the finite-state condition, implies that each of them moves at an essentially constant rate, in the following sense.

\begin{observation}\label{obs:speed}
    For each trailing head $i$, there is a constant $\sigma_i\in[0,1]$, which we call the \emph{speed} of head $i$, such that whenever the leading head is at position $n\in\N$, the position of head $i$ is between $\sigma_i n - |T|$ and $\sigma_i n + |T|$, inclusive.
\end{observation}
\begin{proof}
    As $T$ is finite, repeated application of $\delta_T$ to $t_0$ must eventually result in a cycle, which will thereafter be repeated forever. Let $L$ be the set of positional states in this cycle.

    Let $D$ be the set of positional states that appeared before entering the cycle. Note that $D$ and $L$ are disjoint. After $n$ steps, the cycle has been repeated $(n - |D|)/|L|$ times. For each trailing head $i$, let $\sigma_i\in[0,1]$ be such that head $i$ advances $\sigma_i|L|$ times during the cycle. Then the position of head $i$ after $n$ steps is bounded below by $\sigma_i |L| \left\lfloor \frac{n - |D|}{|L|}\right\rfloor \geq \sigma_i n - |D| - |L| \geq \sigma_i n - |T|$
    and above by $|D| + \sigma_i |L| \left\lceil\frac{n - |D|}{|L|}\right\rceil \leq \sigma_i n + |D| + |L| \leq \sigma_i n + |T|$.
\end{proof}

\section{Multi-Head Finite-State Dimensions}\label{sec:mfsd}
We now define $h$-head finite-state (strong) predimension---which generalize the finite-state dimension of~\cite{DDLM04} and the finite-state strong dimension of~\cite{AHLM07}---and multi-head finite-state (strong) dimension. Definitions~\ref{hydra martingale}--\ref{hdim} continue to follow the standard template for effectivizing fractal dimensions~\cite{DCC,DISS,AHLM07}, differing from the corresponding definitions in~\cite{DDLM04,AHLM07} only in that the gamblers now have $h$ heads.

Briefly, the approach is to define the martingale and $s$-gales induced by the betting function of a specific gambler, then consider the range of parameters $s$ for which it is possible for a gambler's $s$-gale to (strongly) succeed on a given set. Recalling that $s$ parameterizes the favorability of the betting environment, this means we are quantifying how unfavorable the betting environment can become before it becomes impossible for an $h$-FSG to succeed on the set. Intuitively, a gambler can succeed in a more unfavorable environment if the set only contains simple, predictable sequences, so such sets will tend to have lower predimension or dimension.

\begin{definition}\label{hydra martingale}
    The martingale of an $h$-FSG $G$ is the function $d_G: \Sigma^* \rightarrow [0, \infty)$ defined recursively by $d_G(\lambda) = c_0$ and
    $d_G(wb) = |\Sigma| d_G(w)\beta(q_{|w|})(b)$,
    where $q_{|w|}$ is defined as in~\eqref{eq:trajectory}.
\end{definition}

By rescaling this martingale, we derive corresponding $s$-gales.

\begin{definition}\label{fsg s-gale}
    For $s \in [0, \infty)$, the $s$-gale of an $h$-FSG $G$ is the function
    $d_G^{(s)}: \Sigma^* \rightarrow [0, \infty)$ defined, for all $w \in \Sigma^*$, by
    $d_G^{(s)}(w) = |\Sigma|^{(s-1)|w|}d_G(w)$.
\end{definition}

Given a set of sequences, we may then consider the range of parameters $s$ for which $h$-FSGs with (strongly) successful $s$-gales exist.

\begin{definition}\label{mathcalG}
    Let $X \subseteq \Sigma^{\omega}$.
    \begin{enumerate}
        \item $\mathcal{G}_h(X)=\left\{s \in [0, \infty)\:\middle|\;\exists h\text{-FSG }G\text{ with }X \subseteq S^{\infty}\big[d_G^{(s)}\big]\right\}$.
        \item $\mathcal{G}^{\strong}_h(X)=\left\{s \in [0, \infty)\:\middle|\;\exists h\text{-FSG }G\text{ with }X \subseteq S_{\strong}^{\infty}\big[d_G^{(s)}\big]\right\}$.
    \end{enumerate}
\end{definition}

The infima of these ranges then give us our predimensions; the $h$-head finite-state (strong) predimension of a set $X$ is the infimum of the parameters $s$ such that an $h$-FSG can have an $s$-gale that (strongly) succeeds on $X$.

\begin{definition}\label{hdim}
    Let $h\in \Z^+$ and $X\subseteq \Sigma^{\omega}$.
    \begin{enumerate}
        \item The \emph{$h$-head finite-state predimension} of $X$ is
        $\dimFS{(h)}(X) = \inf \mathcal{G}_{h}(X)$.
        \item The \emph{$h$-head finite-state strong predimension} of $X$ is 
        $\DimFS{(h)}(X) = \inf \mathcal{G}_{h}^{\strong}(X)$.
    \end{enumerate}
\end{definition}

Note that $\dimFS{(1)}$ is the finite-state dimension $\dimFS{}$ defined in~\cite{DDLM04} and $\DimFS{(1)}$ is the finite-state strong dimension $\DimFS{}$ defined in~\cite{AHLM07}.

\begin{definition}\label{mhdim}
    Let $X\subseteq\Sigma^{\omega}$.
    \begin{enumerate}
        \item The \emph{multi-head finite-state dimension} of $X$ is $\dimFS{\mh}(X) = \inf_{h\in\Z^+}\dimFS{(h)}(X)$.
        \item The \emph{multi-head finite-state strong dimension} of $X$ is $\DimFS{\mh}(X) = \inf_{h\in\Z^+}\DimFS{(h)}(X)$.
    \end{enumerate}
\end{definition}

Both $h$-head finite-state predimension and multi-head finite-state dimension are monotone, in the following sense.
\begin{observation}\label{obs:mhmonotone}
    Let $X, Y \subseteq \Sigma^{\omega}$.
    \begin{enumerate}
        \item For all $h\in \Z^+$, if $X \subseteq Y$, then $\dimFS{(h)}(X) \le \dimFS{(h)}(Y)$ and $\DimFS{(h)}(X) \le \DimFS{(h)}(Y)$.
        \item If $X \subseteq Y$, then $\dimFS{\mh}(X) \le \dimFS{\mh}(Y)$ and $\DimFS{\mh}(X) \le \DimFS{\mh}(Y)$.
    \end{enumerate}
\end{observation}
\begin{proof}
    Let $h\in\Z^+$ and assume $X\subseteq Y$. By Definition~\ref{mathcalG}, for each $s\in\mathcal{G}_h(Y)$ there is some $h$-FSG $G$ such that $X\subseteq Y\subseteq S^\infty\big[d_G^{(s)}\big]$, so $s\in\mathcal{G}_h(X)$. Hence, $\mathcal{G}_h(Y) \subseteq \mathcal{G}_h(X)$, and therefore $\dimFS{(h)}(X) \le \dimFS{(h)}(Y)$, by Definition~\ref{hdim}. $\DimFS{(h)}(X) \le \DimFS{(h)}(Y)$ holds by the same argument.
      
    By Definition~\ref{mhdim}, the second statement follows immediately from the first.
\end{proof}

\begin{definition}
    The \emph{$h$-head finite-state predimension}, \emph{$h$-head finite-state strong predimension}, \emph{multi-head finite-state dimension}, and \emph{multi-head finite-state strong dimension} of an individual sequence $S\in\Sigma^\omega$ are the $h$-head finite-state predimension, $h$-head finite-state strong predimension, multi-head finite-state dimension, and multi-head finite-state strong dimension, respectively, of the singleton $\{S\}$.
\end{definition}

\section{A Hierarchy of Predimensions}\label{sec:hierarchy}
In this section, we prove our main theorem, a strict hierarchy on predimensions, showing that additional heads enhance the predictive power of multi-head finite-state gamblers.
\begin{theorem}\label{thm:main}
    For each $h\in\Z^+$, there is a sequence $Y$ such that
    \[\dimFS{(h+1)}(Y)<\dimFS{(h)}(Y)\quad\text{and}\quad\DimFS{(h+1)}(Y)<\DimFS{(h)}(Y).\]
\end{theorem}

We will prove this by defining an explicit family of functions on sequences, which we will then apply to Martin-L\"of random sequences. The goal, for each $h\in\Z^+$, is to define sequences in which a constant fraction of the bits depend on exactly $h$ well-spaced bits from earlier in the sequence, so that $h$ trailing heads are necessary and sufficient to access the relevant earlier bits.

\begin{definition}
    For each $i\in\Z^+$, let $p_i$ denote the $i$\textsuperscript{th} prime number: $p_1=2$, $p_2=3$, etc. For each $x\in\mathbb{N}$, define the \emph{multiplicity} of $p_i$ in $x$ to be $\nu_i(x)=\max\big\{k\in\mathbb{N}\;\big|\;p_i^k\mid x\big\}$.
\end{definition}

\begin{definition}\label{functionp1} 
For each $h \in \Z^+$, define the function $F_{h+1}: \C \rightarrow \C$ as follows. For all $S \in \C$, let $F_{h+1}(S)[0]=S[0]$ and, for all $q \in \N$ and $0 \le r < p_{h+1}$, let
\[
    F_{h+1}(S)[qp_{h+1}+r]= 
    \begin{cases}
        S[q(p_{h+1}-1) + r]& \text{if } r>0\\
        \underset{k=1}{\overset{h}{\bigoplus}} F_{h+1}(S)[qp_k]& \text{if }r=0, 
    \end{cases}
\]
where $\oplus$ is the parity operation.
\end{definition}
For example, in the sequence $F_3(S)$, four out of every five bits are ``new'' bits taken directly from $S$, and the fifth bit, at some index $5q$, is given by the parity of the bits at indices $2q$ and $3q$ in $F_3(S)$. Those bits themselves might come directly from $S$ or, if $q$ is a multiple of 5, the recursive dependence may continue. Note that an index in $S$ may be reachable by an even number of paths in the recursion tree whose root is at $5q$, in which case cancellation makes $F_3(S)[5q]$ insensitive to the bit at that index. In Figure~\ref{fig:f3} we show an example of the depth-2 recursive calculation of $F_3(S)[150]$, in which $S[29]$ cancels itself out.

\tikzset{every picture/.style={line width=0.75pt}}    

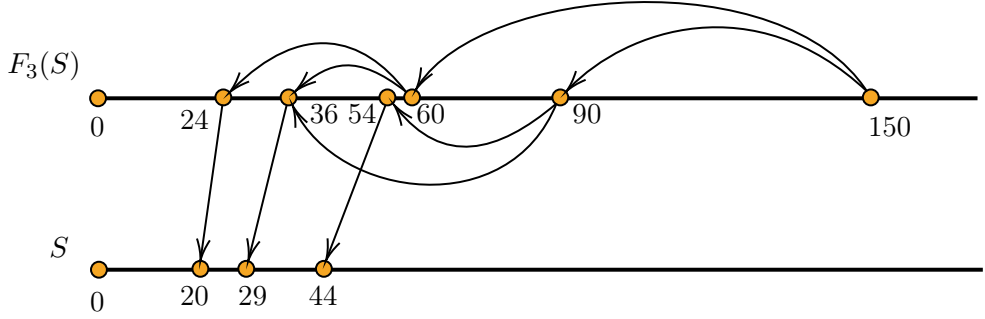
\begin{figure}
    \tikzset{every picture/.style={line width=0.75pt}} 

\begin{tikzpicture}[x=0.75pt,y=0.75pt,yscale=-1,xscale=1]

\draw [line width=1.5]    (99.55,108.86) -- (545.33,108.86) ;
\draw [line width=1.5]    (105.78,195.38) -- (548,195.38) ;
\draw  [fill={rgb, 255:red, 245; green, 166; blue, 35 }  ,fill opacity=1 ] (172.52,194.93) .. controls (172.52,197.15) and (174.31,198.94) .. (176.52,198.94) .. controls (178.73,198.94) and (180.52,197.15) .. (180.52,194.93) .. controls (180.52,192.71) and (178.73,190.92) .. (176.52,190.92) .. controls (174.31,190.92) and (172.52,192.71) .. (172.52,194.93) -- cycle ;
\draw  [fill={rgb, 255:red, 245; green, 166; blue, 35 }  ,fill opacity=1 ] (487.5,108.41) .. controls (487.5,110.63) and (489.29,112.43) .. (491.5,112.43) .. controls (493.71,112.43) and (495.5,110.63) .. (495.5,108.41) .. controls (495.5,106.2) and (493.71,104.4) .. (491.5,104.4) .. controls (489.29,104.4) and (487.5,106.2) .. (487.5,108.41) -- cycle ;
\draw  [fill={rgb, 255:red, 245; green, 166; blue, 35 }  ,fill opacity=1 ] (97.78,108.86) .. controls (97.78,111.07) and (99.57,112.87) .. (101.78,112.87) .. controls (103.99,112.87) and (105.78,111.07) .. (105.78,108.86) .. controls (105.78,106.64) and (103.99,104.84) .. (101.78,104.84) .. controls (99.57,104.84) and (97.78,106.64) .. (97.78,108.86) -- cycle ;
\draw  [fill={rgb, 255:red, 245; green, 166; blue, 35 }  ,fill opacity=1 ] (330.9,108.41) .. controls (330.9,110.63) and (332.69,112.43) .. (334.9,112.43) .. controls (337.11,112.43) and (338.9,110.63) .. (338.9,108.41) .. controls (338.9,106.2) and (337.11,104.4) .. (334.9,104.4) .. controls (332.69,104.4) and (330.9,106.2) .. (330.9,108.41) -- cycle ;
\draw  [fill={rgb, 255:red, 245; green, 166; blue, 35 }  ,fill opacity=1 ] (256.15,108.41) .. controls (256.15,110.63) and (257.95,112.43) .. (260.16,112.43) .. controls (262.37,112.43) and (264.16,110.63) .. (264.16,108.41) .. controls (264.16,106.2) and (262.37,104.4) .. (260.16,104.4) .. controls (257.95,104.4) and (256.15,106.2) .. (256.15,108.41) -- cycle ;
\draw  [fill={rgb, 255:red, 245; green, 166; blue, 35 }  ,fill opacity=1 ] (243.7,108.41) .. controls (243.7,110.63) and (245.49,112.43) .. (247.7,112.43) .. controls (249.91,112.43) and (251.71,110.63) .. (251.71,108.41) .. controls (251.71,106.2) and (249.91,104.4) .. (247.7,104.4) .. controls (245.49,104.4) and (243.7,106.2) .. (243.7,108.41) -- cycle ;
\draw  [fill={rgb, 255:red, 245; green, 166; blue, 35 }  ,fill opacity=1 ] (193.87,108.41) .. controls (193.87,110.63) and (195.66,112.43) .. (197.87,112.43) .. controls (200.09,112.43) and (201.88,110.63) .. (201.88,108.41) .. controls (201.88,106.2) and (200.09,104.4) .. (197.87,104.4) .. controls (195.66,104.4) and (193.87,106.2) .. (193.87,108.41) -- cycle ;
\draw  [fill={rgb, 255:red, 245; green, 166; blue, 35 }  ,fill opacity=1 ] (160.95,108.41) .. controls (160.95,110.63) and (162.74,112.43) .. (164.95,112.43) .. controls (167.16,112.43) and (168.96,110.63) .. (168.96,108.41) .. controls (168.96,106.2) and (167.16,104.4) .. (164.95,104.4) .. controls (162.74,104.4) and (160.95,106.2) .. (160.95,108.41) -- cycle ;
\draw  [fill={rgb, 255:red, 245; green, 166; blue, 35 }  ,fill opacity=1 ] (149.38,194.93) .. controls (149.38,197.15) and (151.17,198.94) .. (153.39,198.94) .. controls (155.6,198.94) and (157.39,197.15) .. (157.39,194.93) .. controls (157.39,192.71) and (155.6,190.92) .. (153.39,190.92) .. controls (151.17,190.92) and (149.38,192.71) .. (149.38,194.93) -- cycle ;
\draw  [fill={rgb, 255:red, 245; green, 166; blue, 35 }  ,fill opacity=1 ] (211.67,194.93) .. controls (211.67,197.15) and (213.46,198.94) .. (215.67,198.94) .. controls (217.88,198.94) and (219.67,197.15) .. (219.67,194.93) .. controls (219.67,192.71) and (217.88,190.92) .. (215.67,190.92) .. controls (213.46,190.92) and (211.67,192.71) .. (211.67,194.93) -- cycle ;
\draw  [fill={rgb, 255:red, 245; green, 166; blue, 35 }  ,fill opacity=1 ] (98.22,195.38) .. controls (98.22,197.59) and (100.01,199.39) .. (102.22,199.39) .. controls (104.44,199.39) and (106.23,197.59) .. (106.23,195.38) .. controls (106.23,193.16) and (104.44,191.36) .. (102.22,191.36) .. controls (100.01,191.36) and (98.22,193.16) .. (98.22,195.38) -- cycle ;
\draw    (263.56,102.11) .. controls (296.13,52.27) and (446.66,39.51) .. (489.61,105.4) ;
\draw [shift={(262.16,104.41)}, rotate = 299.58] [color={rgb, 255:red, 0; green, 0; blue, 0 }  ][line width=0.75]    (10.93,-3.29) .. controls (6.95,-1.4) and (3.31,-0.3) .. (0,0) .. controls (3.31,0.3) and (6.95,1.4) .. (10.93,3.29)   ;
\draw    (339.6,103.6) .. controls (368.13,74.15) and (432.93,52.08) .. (489.61,105.4) ;
\draw [shift={(337.9,105.41)}, rotate = 312.29] [color={rgb, 255:red, 0; green, 0; blue, 0 }  ][line width=0.75]    (10.93,-3.29) .. controls (6.95,-1.4) and (3.31,-0.3) .. (0,0) .. controls (3.31,0.3) and (6.95,1.4) .. (10.93,3.29)   ;
\draw    (169.59,104.37) .. controls (204.38,69.4) and (235.35,77.11) .. (257,105.5) ;
\draw [shift={(168,106)}, rotate = 313.89] [color={rgb, 255:red, 0; green, 0; blue, 0 }  ][line width=0.75]    (10.93,-3.29) .. controls (6.95,-1.4) and (3.31,-0.3) .. (0,0) .. controls (3.31,0.3) and (6.95,1.4) .. (10.93,3.29)   ;
\draw    (203.01,103.73) .. controls (218.94,90.17) and (232.06,86.06) .. (257,105.5) ;
\draw [shift={(201.27,105.25)}, rotate = 318.61] [color={rgb, 255:red, 0; green, 0; blue, 0 }  ][line width=0.75]    (10.93,-3.29) .. controls (6.95,-1.4) and (3.31,-0.3) .. (0,0) .. controls (3.31,0.3) and (6.95,1.4) .. (10.93,3.29)   ;
\draw    (200.86,114.52) .. controls (223.39,159.94) and (308.98,171.18) .. (333,111.46) ;
\draw [shift={(199.87,112.43)}, rotate = 66.1] [color={rgb, 255:red, 0; green, 0; blue, 0 }  ][line width=0.75]    (10.93,-3.29) .. controls (6.95,-1.4) and (3.31,-0.3) .. (0,0) .. controls (3.31,0.3) and (6.95,1.4) .. (10.93,3.29)   ;
\draw    (252.33,113.83) .. controls (275.52,143.05) and (303.13,136.92) .. (333,111.46) ;
\draw [shift={(250.91,112)}, rotate = 53.1] [color={rgb, 255:red, 0; green, 0; blue, 0 }  ][line width=0.75]    (10.93,-3.29) .. controls (6.95,-1.4) and (3.31,-0.3) .. (0,0) .. controls (3.31,0.3) and (6.95,1.4) .. (10.93,3.29)   ;
\draw    (153.66,188.95) -- (164.51,111.98) ;
\draw [shift={(153.39,190.93)}, rotate = 278.02] [color={rgb, 255:red, 0; green, 0; blue, 0 }  ][line width=0.75]    (10.93,-3.29) .. controls (6.95,-1.4) and (3.31,-0.3) .. (0,0) .. controls (3.31,0.3) and (6.95,1.4) .. (10.93,3.29)   ;
\draw    (178,188.99) -- (196.87,112.43) ;
\draw [shift={(177.52,190.93)}, rotate = 283.85] [color={rgb, 255:red, 0; green, 0; blue, 0 }  ][line width=0.75]    (10.93,-3.29) .. controls (6.95,-1.4) and (3.31,-0.3) .. (0,0) .. controls (3.31,0.3) and (6.95,1.4) .. (10.93,3.29)   ;
\draw    (217.38,189.06) -- (246.7,112.43) ;
\draw [shift={(216.67,190.93)}, rotate = 290.93] [color={rgb, 255:red, 0; green, 0; blue, 0 }  ][line width=0.75]    (10.93,-3.29) .. controls (6.95,-1.4) and (3.31,-0.3) .. (0,0) .. controls (3.31,0.3) and (6.95,1.4) .. (10.93,3.29)   ;

\draw (488.68,116.86) node [anchor=north west][inner sep=0.75pt]   [align=left] {150};
\draw (54.64,83.8) node [anchor=north west][inner sep=0.75pt]   [align=left] {$\displaystyle F_{3}( S)$};
\draw (75.81,177.46) node [anchor=north west][inner sep=0.75pt]   [align=left] {$\displaystyle S$};
\draw (96.28,116.86) node [anchor=north west][inner sep=0.75pt]   [align=left] {0};
\draw (96.28,205.16) node [anchor=north west][inner sep=0.75pt]   [align=left] {0};
\draw (339.64,110.17) node [anchor=north west][inner sep=0.75pt]   [align=left] {90};
\draw (260.89,110.17) node [anchor=north west][inner sep=0.75pt]   [align=left] {60};
\draw (226.19,110.62) node [anchor=north west][inner sep=0.75pt]   [align=left] {54};
\draw (207.06,110.17) node [anchor=north west][inner sep=0.75pt]   [align=left] {36};
\draw (141.66,113.29) node [anchor=north west][inner sep=0.75pt]   [align=left] {24};
\draw (141.66,202.49) node [anchor=north west][inner sep=0.75pt]   [align=left] {20};
\draw (171.02,202.49) node [anchor=north west][inner sep=0.75pt]   [align=left] {29};
\draw (206.61,202.49) node [anchor=north west][inner sep=0.75pt]   [align=left] {44};

\end{tikzpicture}
    \caption{$F_3(S)[150]=S[20]\oplus S[29] \oplus S[29]\oplus S[44]=S[20]\oplus S[44]$.}
    \label{fig:f3}
\end{figure}

\begin{lemma}\label{lem:dim}
   For all $S\in\C$ and $h \in \mathbb{Z}^+$,
   \[\dimFS{(h+1)}(F_{h+1}(S)) \le \DimFS{(h+1)}(F_{h+1}(S))\le 1-\frac{1}{p_{h+1}}.\]
\end{lemma}
\begin{proof}[Proof sketch]
    For each $h\in\Z^+$, we construct an $(h+1)$-head finite-state gambler $G$ such that, for all $\epsilon>0$ and all $S\in\C$, the sequence $F_{h+1}(S)$ is in the success set of the $(1-\frac{1}{p_{h+1}}+\epsilon)$-gale of $G$. In each round of $p_{h+1}$ steps, the $i$\textsuperscript{th} head of $G$ advances $p_i$ times, such that the trailing heads point to the bits referenced in the next parity calculation. This means that $G$ can predict one out of every $p_{h+1}$ bits with certainty, and it bets uniformly on all other bits. This corresponds to its $(1-\frac{1}{p_{h+1}}+\epsilon)$-gale amassing unbounded capital.

    A complete proof is in Appendix~\ref{app:hierarchy}.
\end{proof} 

Fix any positive integer $h$ and any Martin-L\"of random sequence $R\in\C$, and let $Y=F_{h+1}(R)$. To prove Theorem~\ref{thm:main}, it suffices to show that the $h$-head finite-state predimension of $Y$ is 1. To do this, our high-level plan is to show that for every $h$-FSG and each sufficiently long prefix $Y[0..n]$ of $Y$, there is a long suffix $Y[b..n]$ that has high conditional Kolmogorov complexity given the information the trailing heads might read while the leading head reads $Y[b..n]$. This high conditional complexity will bound the rate of growth of the $h$-FSG's martingale on that suffix, and we will apply this algorithmic information-theoretic argument iteratively to bound the martingale's overall rate of growth.

\medskip

Let $G$ be any $h$-FSG over the alphabet $\{0,1\}$. Let $\sigma_1,\ldots,\sigma_{h-1}$ be the trailing head speeds guaranteed by Observation~\ref{obs:speed}, and let $c=|T|$. Without loss of generality, assume $0<\sigma_1< \ldots<\sigma_{h-1}<1$.

We now introduce some additional notation that will be helpful in carefully analyzing the recursive structure of the sequence $Y$. Let $n\in\N$. For every set $A\subseteq[0,n]$, define the string $Y[A]\in\{0,1\}^{n+1}$ by
\[Y[A][i]=\begin{cases}Y[i]&\text{if }i\in A\\0&\text{otherwise,}\end{cases}\]
for all $i\in\{0,\ldots,n\}$. Note that $Y[A]=Y[A\cap\N]$; we will only use integer indices of strings and sequences, but it will sometimes be convenient to consider the real intervals those integers belong to.

Let $b=\lfloor\gamma n\rfloor+1$, where $\gamma\in(0,1) \cap \Q$ is a parameter whose value will be defined in Lemma~\ref{lem:disjoint}. Define $u_h=Y[b..n]$,
\[U=\bigcup_{i=1}^{h-1}[\sigma_ib-c,\sigma_in+c],\]
and $u=Y[U]$. Intuitively, the string $u$ contains all bits that might be read by any trailing head while the leading head $h$ reads $u_h$.

Let $d\in\mathbb{N}$ be a constant that does not depend on $n$. For each $1\leq i\leq h$, define
\[V_i=\left\{(p_i/p_{h+1})^{\nu_{h+1}(m)}m \;\middle|\; m\in[b..n]\text{ and }\nu_{h+1}(m)\leq d\right\}\]
and its \emph{closure}
\[\overline{V}_i=\left\{m\cdot \frac{p_{h+1}^{a_1+\ldots+a_h}}{p_1^{a_1}\cdots p_h^{a_h}}\;\middle|\; m\in V_i\text{ and }a_1,\ldots,a_h\in\N\right\}.\] 

Note that $V_i=\overline{V}_i\cap\{m\mid \nu_{h+1}(m)=0\}$. Define $v_i=Y[V_i]$ and $\overline{v}_i=Y[\overline{V}_i]$.

We prove the following four lemmas in Appendix~\ref{app:hierarchy}. Informally, Lemma~\ref{lem:disjoint} shows that we can tune parameters such that $u$ excludes one of the $h$ regions relevant to calculating parity bits in the suffix $u_h$; Lemma~\ref{lem:3inequ} shows the algorithmic information consequence that $u_h$ has high complexity conditioned on $u$; Lemma~\ref{lem:dtok} shows that when this conditional Kolmogorov complexity is high, a gambler cannot win quickly while betting on the suffix $u_h$; and Lemma~\ref{lem:dbound} applies this reasoning inductively to show that a gambler cannot win much on any prefix of $Y$.
\begin{lemma}\label{lem:disjoint}
    There are constants (depending on $c$, $d$, $h$, and the trailing head speeds, but not on $n$) $\gamma\in(1/2,1)\cap \Q$, $n_0\in\N$, and $1\leq j\leq h$ such that if $n\geq n_0$, then $U\cap \overline{V}_j=\emptyset$.
\end{lemma}
\begin{lemma}\label{lem:3inequ}
Fix $\gamma$ and $j$ as in Lemma~\ref{lem:disjoint}, let $n\in\N$ be sufficiently large, and let $b=\lfloor\gamma n\rfloor+1$. Then $K(u_h\mid u)\ge (n-b)(1-p_{h+1}^{-d-1}) - O(\log n)$.
\end{lemma}

\begin{lemma}\label{lem:dtok}
    Let $h\in\Z^+$, $G$ an $h$-FSG, $\varepsilon>0$ a rational constant, and $n\in\N$ sufficiently large. Define $b$, $u$, and $u_h$ as above. Let $x=u_h[b..n]$ and $y=Y[0..b-1]$. If
    $K(u_h\mid u)\geq (n-b)(1-\frac{\varepsilon}{2})$,
    then
    \begin{equation}\label{eq:nowin}
        \max_{w\sqsubseteq x}\frac{d_G(yw)}{d_G(y)}\leq 2^{\varepsilon\cdot(n-b)}.
    \end{equation}
\end{lemma} 

\begin{lemma}\label{lem:dbound}
    Let $Y = F_{h+1}(R)$, where $R$ is a Martin-L\"of random sequence and $h \in \mathbb{Z}^+$. For all $\epsilon \in \Q \cap (0,1)$ and every $h$-head finite-state gambler $G$, the set
    $\big\{d_G^{(1-2\epsilon)}(Y[0..n-1])\;\big|\; n \in \N \big\}$
    is bounded by a constant.
\end{lemma}

With these lemmas, we prove Theorem~\ref{thm:main}. Letting $\epsilon$ approach 0 and applying Lemma~\ref{lem:dbound} shows that $\dimFS{(h)}(Y)=\DimFS{(h)}(Y)=1$,
while Lemma~\ref{lem:dim} shows that
\[\dimFS{(h+1)}(Y)\leq\DimFS{(h+1)}(Y)\leq 1-\frac{1}{p_{h+1}}.\]
Hence, for every $h \in \mathbb{Z}^+$, $Y= F_{h+1}(R)$ attests to a strict separation between $h$-head and $(h+1)$-head finite-state predimension, and between $h$-head and $(h+1)$-head finite-state strong predimension.

\section{Instability of \texorpdfstring{$h$}{h}-Head Finite-State Predimensions}\label{sec:instability}
We now show that unlike typical fractal dimension notions, $h$-head finite-state predimensions are not stable for fixed $h\geq 2$.
\begin{theorem}\label{thm:unstable}
    For each integer $h\geq 2$, $h$-head finite-state predimension and strong predimension are not stable under finite unions.
\end{theorem}

We prove this using two variants of our function family $F_{h+1}$, called $F'_{h+1}$ and $F''_{h+1}$. Recall that the bit in one out of every $h+1$ positions in a sequence $F_{h+1}(S)$ can be computed as the parity of $h$ previous bits in the sequence. At those same positions in $F'_{h+1}(S)$ and $F''_{h+1}(S)$, the bit can be computed as the parity of only $h-1$ of the $h$ bits involved in the $F_{h+1}(S)$ calculation. In $F'_{h+1}(S)$, it is parity of the first $h-1$ of those $h$ bits, while in $F''_{h+1}(S)$, it is the parity of the last $h-1$ of those $h$ bits. Intuitively, these sequences together capture the full complexity of $F_{h+1}(S)$, but each of them individually is well-suited to using only $h-1$ trailing heads.

Formally, for each $h\in\Z^+$, define the functions $F'_{h+1}:\C\to\C$ and $F''_{h+1}:\C\to\C$ as follows. For all $S\in\C$, let $F'_{h+1}(S)[0]=F''_{h+1}(S)[0]=S[0]$ and, for all $q \in \N$ and $0 \le r < p_{h+1}$, let
    \[
    F'_{h+1}(S)[qp_{h+1}+r]= 
\begin{cases}
    S[q(p_{h+1}-1) + r]& \text{if } r>0\\
    \underset{k=1}{\overset{h-1}{\bigoplus}} F'_{h+1}(S)[qp_k] & \text{if }r=0 \\
    
\end{cases}
\] and 
\[
    F''_{h+1}(S)[qp_{h+1}+r]= 
\begin{cases}
    S[q(p_{h+1}-1) + r]& \text{if }r>0\\
    \underset{k=2}{\overset{h}{\bigoplus}} F''_{h+1}(S)[qp_k] & \text{if }r=0.  
\end{cases}
\]

For any fixed $h$ and $\epsilon>0$, we can construct a pair of $h$-FSGs whose $(1-\frac{1}{p_{h+1}}+\epsilon)$-gales succeed, for all $S\in\C$, on $F'_{h+1}(S)$ and on $F''_{h+1}(S)$, respectively, using constructions almost identical to the $(h+1)$-FSG that succeeds on $F_{h+1}(S)$ in Lemma~\ref{lem:dim}; essentially, we omit either the first or the last trailing head from that construction. This establishes, for all $S\in\C$, that the singletons $\{F'_{h+1}(S)\}$ and $\{F''_{h+1}(S)\}$ have both $h$-head finite-state predimension and strong predimension at most $1-\frac{1}{p_{h+1}}$.

Nevertheless, we show---using techniques very similar to those of our strict hierarchy result, Theorem~\ref{thm:main}---that no $h$-FSG has a $(1-2\epsilon)$-gale that succeeds on \emph{both} $F'_{h+1}(R)$ and $F''_{h+1}(R)$ when $R$ is any Martin-L\"of random sequence. This implies that the set $\{F'_{h+1}(R),F''_{h+1}(R)\}$ has $h$-head finite-state predimension and strong predimension 1, proving instability.

At a high level, the argument is that by the obliviousness of the head movements, the heads of any $h$-FSG must move identically on both these sequences, which means that for at least one of these sequences, at least one of the relevant $h-1$ bits in any parity calculation must be missing from the information the gambler can readily access. Much as in Lemmas~\ref{lem:3inequ},~\ref{lem:dtok}, and~\ref{lem:dbound}, we show that this missing information implies that at least one of the sequences can be broken into strings of high conditional Kolmogorov complexity given the trailing heads' information, which bounds the corresponding gale's rate of growth on each of these strings and thereby on the entire sequence. Our full proof appears in Appendix~\ref{app:instability}.

\section{Stability of Multi-Head Finite-State Dimension}\label{sec:stability}

In this section we prove that, notwithstanding the instability of $h$-head finite-state predimension and strong predimension, multi-head finite-state dimension and strong dimension are stable under finite unions. This is because, as we now prove, a multi-head finite-state gambler with sufficiently many heads can implement a martingale that \emph{approximately averages} the martingales of a pair of other multi-head finite-state gamblers, thereby succeeding wherever either of the pair would succeed, up to $\epsilon$ error. Maintaining this approximate average with finitely many states is delicate; a full proof of the following stability theorem appears in Appendix~\ref{app:stability}.
\begin{theorem}\label{thm:appd}
    Given $h_1, h_2 \in \Z^+$, for each $h_1$-FSG $G_1$, each $h_2$-FSG $G_2$ and each $\epsilon \in (0,1)$, there is a $(h_1+h_2-1)$-FSG $G$ such that for all $s \in [0,1]$,
    \[S^{\infty}\big[d_{G_1}^{(s)}\big] \cup S^{\infty}\big[d_{G_2}^{(s)}\big] \subseteq S^{\infty}\big[d_{G}^{(s+\epsilon)}\big]\quad\text{and}\quad S_{\strong}^{\infty}\big[d_{G_1}^{(s)}\big] \cup S_{\strong}^{\infty}\big[d_{G_2}^{(s)}\big] \subseteq S_{\strong}^{\infty}\big[d_{G}^{(s+\epsilon)}\big].\]
\end{theorem}

\begin{corollary}\label{cor:stable}
For all $X,Y\subseteq\Sigma^\omega$,
\[\dimFS{\mh}(X \cup Y) = \max \{\dimFS{\mh}(X), \dimFS{\mh}(Y)\}\]
and 
\[\DimFS{\mh}(X \cup Y) = \max \{\DimFS{\mh}(X), \DimFS{\mh}(Y)\}.\]
\end{corollary}
\begin{proof}
It is immediate from Theorem~\ref{thm:appd} that $\dimFS{\mh}(X \cup Y) \le \max \{\dimFS{\mh}(X), \dimFS{\mh}(Y)\}$ and $\DimFS{\mh}(X \cup Y) \le \max \{\DimFS{\mh}(X), \DimFS{\mh}(Y)\}$. The reverse inequalities follow from Observation~\ref{obs:mhmonotone}.
\end{proof}

\section{Conclusion}

To our knowledge, the finite-state dimensions $\dimFS{}$ and $\DimFS{}$ are the most effective (meaning effectivized in the most restrictive way) of the fractal dimensions that have been investigated. It is not hard to see that the multi-head finite-state dimensions $\dimFS{\mh}$ and $\DimFS{\mh}$, while not as effective as the finite-state dimensions, are \emph{highly effective} in the sense that they are more effective than the polynomial-time dimensions $\dim_\textup{p}$ and $\Dim_\textup{p}$ of~\cite{DCC,AHLM07}. 

Other dimensions that are highly effective in the above sense have been investigated. These include pushdown dimensions, Lempel--Ziv dimensions, and perhaps polylogspace dimensions~\cite{AlMaMo2017,AMMP2008,DotNic2007,Lope2006,Nich2004,MaMoPe2011}. It will be interesting to see how multi-head finite-state dimensions are related to these other highly effective dimensions.

It will be noted that our motivating question remains open:  A sequence $S \in\Sigma^\omega$ is \emph{multi-head finite-state random} if there is no multi-head finite-state gambler $G$ such that $S \in S^\infty[d_G]$.  Is every sequence with multi-head finite-state dimension 1 multi-head finite-state random?

\bibliography{mhd}

\appendix
\section{Appendix: Proofs from Section~\ref{sec:hierarchy}}\label{app:hierarchy}

\begin{proof}[Proof of Lemma~\ref{lem:dim}]
    Let $S\in\C$, $h \in \mathbb{Z}^+$, and $Y=F_{h+1}(S)$. We define an $(h+1)$-head finite-state gambler $G = (T \times Q, \Sigma, \delta, \mu, \beta, (t_0, q_0), c_0)$ as follows:
    \begin{itemize}
        \item $\Sigma = \{0,1\}.$
        \item $|T\times Q| = 2p_{h+1}-1, T=\{t_0, t_1, t_2,\cdots, t_{p_{h+1}}\}, Q = \{q_0, q_1, \cdots, q_{2p_{h+1}-2}\}$.
        \item $\mu(t_0) = \mu(t_1) = \mu(t_{p_{h+1}}) = \{\overbrace{1\;\cdots\;1}^{h}\}$;
        
        For all integers $i, j, \text{ where } 2 \le i \le h+1, \;p_{i-1}-1<j\le p_i-1, \mu(t_j)=\{\overbrace{0\;\cdots\;0}^{i-1}\;1\;\cdots\;1\}$.
        \item $\forall 0 \le i \le p_{h+1}-1, \delta_T(t_i) = t_{i+1}$ and $\delta_T(t_{p_{h+1}}) = t_0.$
        \item For all $\vec{\sigma} \in \Sigma^{h+1}, \;  \delta_Q(q_0, \vec{\sigma}) = q_1, \delta_Q(q_1, \vec{\sigma}) = q_2$; 
        
        For all $2 < i \le 2p_{h+1}-2$, $q_i$ records the current XOR result, which is $0$ if $i$ is even, $1$ else.
        
        Let $\vec{\sigma}[i]$ be the $i$\textsuperscript{th} element in the vector $\vec{\sigma}$,   
        \[
    \delta_Q(q_2, \vec{\sigma}) = 
    \begin{cases}
    q_3, & \text{if } \vec{\sigma}[1]=1,\\
    q_4,  & \text{if } \vec{\sigma}[1]=0.
    \end{cases}
    \]
    $\forall 2 \le i \le h,\text{ Let } j= 2(p_i-1)$, then
    \[
    \delta_Q(q_j, \vec{\sigma}) = 
    \begin{cases}
    q_{j+2}, & \text{if } \vec{\sigma}[i]=0,\\
    q_{j+1},  & \text{if } \vec{\sigma}[i]=1.
    \end{cases}
    \]
    and 
    \[
    \delta_Q(q_{j-1}, \vec{\sigma}) = 
    \begin{cases}
    q_{j+2}, & \text{if } \vec{\sigma}[i]=1,\\
    q_{j+1},  & \text{if } \vec{\sigma}[i]=0.
    \end{cases}
    \]
    $\forall \vec{\sigma} \in \Sigma^{h+1}, \;  \delta_Q(q_{2p_{h+1}-2}, \vec{\sigma}) = \delta_Q(q_{2p_{h+1}-3}, \vec{\sigma})=q_1$.
    
    For all other states, $\forall \vec{\sigma} \in \Sigma^{h+1}$, $\delta_Q(q_i, \vec{\sigma}) = q_{i+2}$.
    \item $\beta(q_{2p_{h+1}-2})(0)=1$, $\beta(q_{2p_{h+1}-2})(1)=0$, $\beta(q_{2p_{h+1}-3})(0)=0$, $\beta(q_{2p_{h+1}-3})(1)=1$; $\forall q \in Q\setminus\{q_{2p_{h+1}-2}, q_{2p_{h+1}-3}\}, \beta(q)(0) = \beta(q)(1) = \frac{1}{2}$.
   \end{itemize}
    
    Based on our construction,
    \[d_G(Y[0..n-1]) = 2^{\lfloor n/p_{h+1} \rfloor}.\]
    Hence, for all $\epsilon \in \Q \cap (0,1)$, the $(1-\frac{1}{p_{h+1}} + \epsilon)$-gale of $G$ strongly succeeds on $Y$:
    \begin{align*}
        d_{G}^{\Big(1-\frac{1}{p_{h+1}} + \epsilon\Big)}(Y[0..n-1]) &= 2^{\Big(1-\frac{1}{p_{h+1}} + \epsilon -1\Big)n}d_{G}(Y[0..n-1]) \\
        & = 2^{\Big(\epsilon -\frac{1}{p_{h+1}}\Big)n}\times 2^{\lfloor \frac{n}{p_{h+1}} \rfloor}\\
        & \ge 2^{\Big(\epsilon -\frac{1}{p_{h+1}}\Big)n}\times 2^{\frac{n}{p_{h+1}} -1} \\
        & = 2^{\epsilon n -1}.
    \end{align*}
    Therefore, \[\dimFS{(h+1)}(Y)\le \DimFS{(h+1)}(Y) \le 1-\frac{1}{p_{h+1}}.\]
\end{proof}

\begin{proof}[Proof of Lemma~\ref{lem:disjoint}]
    First, to guarantee $b > \sigma_{h-1} n+c$, we will require $\gamma > \sigma_{h-1} + \frac{c-1}{n}$.

    Let $\epsilon = \frac{1-\sigma_{h-1}}{4}$, noting that $\sigma_{h-1} + \epsilon < 1$. Choose $n$ so that $\frac{c-1}{n} < \epsilon.$

    Define $\overline{T}_i$ as
    \[\overline{T}_i=\left\{\Big(\frac{p_i}{p_{h+1}}\Big)^{k'}\cdot \frac{p_{h+1}^{a_1+\ldots+a_h}}{p_1^{a_1}\cdots p_h^{a_h}} \in (0,1)\;\middle|\; k'\in [0,d]\right\}.\]
    
    Define \[T' = \bigcup _{i = 0}^{h-1} \overline{T}_i,\] then $\lvert T' \rvert < \infty$. Then $\exists \zeta, \forall \tau_a, \tau_b \in T', \lvert \tau_a - \tau_b \rvert > \zeta.$

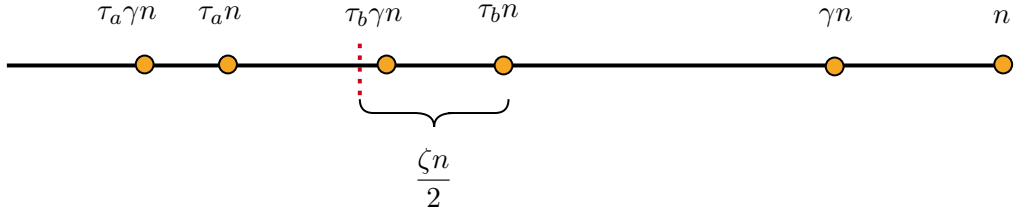
\tikzset{every picture/.style={line width=0.75pt}}   

\begin{tikzpicture}[x=0.75pt,y=0.75pt,yscale=-1,xscale=1]

\draw [line width=1.5]    (67,146) -- (568,146) ;
\draw  [fill={rgb, 255:red, 245; green, 166; blue, 35 }  ,fill opacity=1 ] (565,145.5) .. controls (565,147.99) and (567.01,150) .. (569.5,150) .. controls (571.99,150) and (574,147.99) .. (574,145.5) .. controls (574,143.01) and (571.99,141) .. (569.5,141) .. controls (567.01,141) and (565,143.01) .. (565,145.5) -- cycle ;
\draw  [fill={rgb, 255:red, 245; green, 166; blue, 35 }  ,fill opacity=1 ] (480,146.5) .. controls (480,148.99) and (482.01,151) .. (484.5,151) .. controls (486.99,151) and (489,148.99) .. (489,146.5) .. controls (489,144.01) and (486.99,142) .. (484.5,142) .. controls (482.01,142) and (480,144.01) .. (480,146.5) -- cycle ;
\draw  [fill={rgb, 255:red, 245; green, 166; blue, 35 }  ,fill opacity=1 ] (313,146) .. controls (313,148.49) and (315.01,150.5) .. (317.5,150.5) .. controls (319.99,150.5) and (322,148.49) .. (322,146) .. controls (322,143.51) and (319.99,141.5) .. (317.5,141.5) .. controls (315.01,141.5) and (313,143.51) .. (313,146) -- cycle ;
\draw  [fill={rgb, 255:red, 245; green, 166; blue, 35 }  ,fill opacity=1 ] (254,145.5) .. controls (254,147.99) and (256.01,150) .. (258.5,150) .. controls (260.99,150) and (263,147.99) .. (263,145.5) .. controls (263,143.01) and (260.99,141) .. (258.5,141) .. controls (256.01,141) and (254,143.01) .. (254,145.5) -- cycle ;
\draw  [fill={rgb, 255:red, 245; green, 166; blue, 35 }  ,fill opacity=1 ] (174,145.5) .. controls (174,147.99) and (176.01,150) .. (178.5,150) .. controls (180.99,150) and (183,147.99) .. (183,145.5) .. controls (183,143.01) and (180.99,141) .. (178.5,141) .. controls (176.01,141) and (174,143.01) .. (174,145.5) -- cycle ;
\draw  [fill={rgb, 255:red, 245; green, 166; blue, 35 }  ,fill opacity=1 ] (132,145.5) .. controls (132,147.99) and (134.01,150) .. (136.5,150) .. controls (138.99,150) and (141,147.99) .. (141,145.5) .. controls (141,143.01) and (138.99,141) .. (136.5,141) .. controls (134.01,141) and (132,143.01) .. (132,145.5) -- cycle ;
\draw   (245,163) .. controls (245,167.67) and (247.33,170) .. (252,170) -- (272.5,170) .. controls (279.17,170) and (282.5,172.33) .. (282.5,177) .. controls (282.5,172.33) and (285.83,170) .. (292.5,170)(289.5,170) -- (313,170) .. controls (317.67,170) and (320,167.67) .. (320,163) ;
\draw [color={rgb, 255:red, 208; green, 2; blue, 27 }  ,draw opacity=1 ][line width=1.5]  [dash pattern={on 1.69pt off 2.76pt}]  (245,135) -- (245,164) ;

\draw (111,115) node [anchor=north west][inner sep=0.75pt]   [align=left] {$\displaystyle \tau _{a} \gamma n$};
\draw (162,115) node [anchor=north west][inner sep=0.75pt]   [align=left] {$\displaystyle \tau _{a} n$};
\draw (236,116) node [anchor=north west][inner sep=0.75pt]   [align=left] {$\displaystyle \tau _{b} \gamma n$};
\draw (303,114) node [anchor=north west][inner sep=0.75pt]   [align=left] {$\displaystyle \tau _{b} n$};
\draw (475,116) node [anchor=north west][inner sep=0.75pt]   [align=left] {$\displaystyle \gamma n$};
\draw (563,117) node [anchor=north west][inner sep=0.75pt]   [align=left] {$\displaystyle n$};
\draw (271,186) node [anchor=north west][inner sep=0.75pt]   [align=left] {$\displaystyle \frac{\zeta n}{2}$};

\end{tikzpicture}

For any pair $\tau_a, \tau_b \in T'$, we can find $\gamma$ such that $[\tau_a\gamma n, \tau_a n]$ and $[\tau_b\gamma n, \tau_b n]$ are disjoint.
Without loss of generality, assume $\tau_a < \tau_b$. Then we can choose $\gamma$ so that
\begin{align*}
    \tau_b\gamma n &> \tau_b n - \frac{\zeta n}{2} \\
    \gamma &> 1 - \frac{\zeta}{2 \tau_b}.
\end{align*}
Hence, we can choose $\gamma > 1 - \frac{\zeta}{2 \tau_{max}}$ where $\tau_{max} = \max\{\tau \mid \tau \in T'\}$ so that $\overline{T}_i$ are disjoint.

\medskip
For all $1 \le i \le h-1$, the length of the interval $[\sigma_ib-c,\sigma_in+c]$ can be expressed as 
\begin{align*}
    \sigma_i n +c - (\sigma_i(\lfloor\gamma n\rfloor +1)-c) & = \sigma_i n - \sigma_i(\lfloor\gamma n\rfloor +1) + 2c \\
    & \le \sigma_i n(1 - \gamma) + 2c.
\end{align*}
If this length is less than $\frac{\zeta}{2}n$, then we can guarantee there exists at most one $\tau \in T$ so that $[\sigma_ib-c,\sigma_in+c] \cap [\tau\gamma n, \tau n] \neq \emptyset$. Rearranging, this condition is satisfied if
\[n \ge \frac{4c}{\zeta - 2\sigma_i(1 - \gamma)}.\]

Therefore, we can select rational
\[\gamma > \max\left\{\sigma_{h-1} + \epsilon, 1 - \frac{\zeta}{2 \tau_{max}}, 1 - \frac{\zeta}{2 \sigma_{h-1}}\right\}\]
and every
\[n_0>\max\left\{\frac{4(c-1)}{1-\sigma_{h-1}}, \frac{4c}{\zeta - 2\sigma_{h-1}(1 - \gamma)}\right\},\]
such that, for all $n\geq n_0$, there is some $1\leq j\le h$ satisfying $U\cap \overline{V}_j=\emptyset$.
\end{proof}

\begin{proof}[Proof of Lemma~\ref{lem:3inequ}]
    Define $W=\{0,\ldots,b-1\}\setminus \overline{V}_j$ and $w=Y[W]$. We will prove the following three inequalities.
    \begin{align}
    K(u_h\mid u)&\ge K(u_h\mid w) - O(1) \label{geq:1st} \\
    &\ge  K(v_j\mid w) - O(\log n) \label{geq:2nd} \\
    &\ge (n-b)(1-p_{h+1}^{-d-1}) - O(\log n). \label{geq:3rd}
\end{align}

    For \eqref{geq:1st}, Lemma~\ref{lem:disjoint} gives $U\cap \overline{V}_j=\emptyset$. Since $W=\{0,\ldots,b-1\}\setminus \overline{V}_j$, we have $U\subseteq W$. It follows that $u$ is computable from $w$, so by~\eqref{eq:soi} and~\eqref{eq:dpi}, $K(u_h\mid u) + O(1) \ge K(u_h\mid w)$.
    \[u[i]=\begin{cases}w[i]&\text{if }i\in U\\0&\text{otherwise.}\end{cases}\]

For \eqref{geq:2nd}, we first prove the following claim.
\begin{claim}\label{Claim:comp}
    $(u_h,w)\mapsto v_j$ is a computable mapping.
\end{claim}
\begin{claimproof}
    For each $k\in\mathbb{N}$, let
    \[V^{(k)}_j=\left\{t\in V_j\mid \nu_j(t)\leq k\right\}\]
    and
    $v^{(k)}_j=Y\big[V^{(k)}_j\big]$. We show by induction on $k$ that, for all $k\in\mathbb{N}$, the mapping $(u_h,w)\mapsto v^{(k)}_j$ is computable. For the basis step, simply observe that $V^{(0)}_j\subseteq[b,n]$.
    
    For the inductive step, fix $k\in\mathbb{N}$, assume $(u_h,w)\mapsto v^{(k)}_j$ is computable, and let $r\in V_j^{(k+1)}$. Then $r=s\cdot p_j^e$ for some $e\in\{0,\ldots,k+1\}$ and $s$ satisfying $\nu_{h+1}(s)=0$ and $s\cdot p_{h+1}^e\in[b,n]$. If $e=0$, then $r\in[b,n]$. Otherwise,
    \[Y[s\cdot p_{h+1}^e]=\bigoplus_{i_1,\ldots,i_e\in\{1,\ldots,h\}}Y\left[s\cdot p_{i_1}\cdots p_{i_e}\right],\]
    which can be rearranged to obtain a formula for $Y[r]=Y[s\cdot p_j^e]$:
    \begin{equation}\label{eq:recover}
        Y[r]=Y[s\cdot p_{h+1}^e]\oplus\bigoplus_{\substack{i_1,\ldots,i_e\in\{1,\ldots,h\}\\(i_1,\ldots,i_e)\neq(j,\ldots,j)}}Y\left[s\cdot p_{i_1}\cdots p_{i_e}\right].
    \end{equation}
    Therefore, it suffices to show that every index appearing in~\eqref{eq:recover} belongs to $[b,n]\cup W\cup V^{(k)}_j$.
    
    The index $s\cdot p_{h+1}^e$ of the first term of~\eqref{eq:recover} belongs to $[b,n]$. Each of the remaining indices can be written in the form $t=s\cdot p_j^{\ell}p_{i_1}\cdots p_{i_{e-\ell}}$, for some $0\leq\ell\leq e-1$ and $i_1,\ldots,i_{e-\ell}\in\{1,\ldots,h\}\setminus\{j\}$. If $t\geq b$, then $t<s\cdot p_{h+1}^e\leq n$ implies $t\in[b,n]$; hence, assume $t<b$. If $t\not\in W$, then $t\in\overline{V}_j$, and it follows that $t\in V_j$ since
    \[\nu_{h+1}(t)=\nu_{h+1}(s\cdot p_j^{\ell}p_{i_1}\cdots p_{i_{e-\ell}})=\nu_{h+1}(s)=0.\]
    In particular, since
    \[\nu_j(t)=\nu_j(s)+\ell\leq \nu_j(s)+e-1\leq \nu_j(r)-1\leq k,\]
    we have $t\in V_j^{(k)}$ in this case. Thus, each bit of $Y$ in~\eqref{eq:recover} can be obtained from $u_h$, $w$, or $v_j^{(k)}$. By induction, we conclude that the mapping is computable.
    \end{claimproof}
    
    Therefore,
\begin{align*}
    K(v_j\mid w)
    &\leq K(v_j,w)-K(w)+O(\log n)\tag{symmetry of information}\\
    &\leq K(u_h,w)-K(w)+O(\log n)\tag{by~\eqref{eq:dpi} and Claim~\ref{Claim:comp}}\\
    &\leq K(u_h\mid w)+O(\log n).\tag{symmetry of information}
\end{align*}
For \eqref{geq:3rd}, first, observe that
\begin{align*}
    |V_j|&=\left|\left\{(p_j/p_{h+1})^{\nu_{h+1}(m)}m \;\middle|\; m\in[b,n]\text{ and }\nu_{h+1}(m)\in[0,d]\right\}\right|\\
    &=\left|\left\{m\in[b,n]\mid\nu_{h+1}(m)\in[0,d]\right\}\right|\\
    &=(n-b+1)-\left|\left\{m\in[b,n]\mid\nu_{h+1}(m)\geq d+1\right\}\right|\\
    &\geq (n-b+1)-\left\lceil\frac{n-b+1}{p_{h+1}^{d+1}}\right\rceil\\
    &\geq (n-b)\left(1-p_{h+1}^{-d-1}\right)-O(1).
\end{align*}
It suffices to show that the bits at these indices are nearly incompressible given $w$, in the sense that $K(v_j\mid w)\geq |V_j|-O(\log n)$.

Let $m=\left\lfloor (p_{h+1}-1) \frac{n}{p_{h+1}}\right\rfloor+(n\bmod p_{h+1})$. Since each index $t\in V_j$ has $\nu_{h+1}(t)=0$, each nontrivial bit of $v_j$ is copied directly from a bit in $R$. Let $A\subseteq[0,m]$ be the set of indices in $R$ of those bits, so there is a computable bijection between $v_j$ and $R[A]$. Furthermore, the set $A$ is computable given $m$ and satisfies $|A|=|V_j|$. Let $B=[0,m]\setminus A$. Since $W$ and $\overline{V}_j$ are disjoint, there is a computable mapping from $R[B]$ to $w$. Furthermore, since $R$ is Martin-L\"of random, we have
\begin{equation}\label{eq:rarb}
    K(R[A],R[B])\geq K(R[0..m])-O(1)\geq m-O(1).
\end{equation}
It follows that
\begin{align*}
    K(v_j\mid w)&\geq K(v_j\mid R[B])-O(1)\tag{$R[B]\mapsto w$ is computable}\\
    &\geq K(R[A]\mid R[B])-O(1)\tag{$v_j\mapsto R[A]$ is computable}\\
    &\geq K(R[A],R[B])-K(R[B])-O(\log n)\tag{symmetry of information}\\
    &= m-K(R[B])-O(\log n)\tag{by~\eqref{eq:rarb}}\\
    &\geq m-|B|-O(\log n)\tag{$R[B]$ is computable from $|B|$ bits}\\
    &=|A|-O(\log n)\\
    &=|V_j|-O(\log n).
\end{align*}
\end{proof}

\begin{proof}[Proof of Lemma~\ref{lem:dtok}]
    First observe that given $u$ and the state $(t_b,q_b)$ of $G$ after the leading head has read $y$, we can compute the function $g:\{0,1\}^{(n-b+1)}\to\Q$ defined by
    \[g(z)=\max_{w\sqsubseteq z}\frac{d_G(yw)}{d_G(y)},\]
    even without knowing $y$. This is because, for each $w\sqsubseteq z$, as $u$ contains all information the trailing heads could possibly access while the leading head's position is in $[b,n]\supseteq [b,b+|w|-1]$, we can simulate $G$ to determine its partial state sequence $((t_{b+1},q_{b+1}),\ldots,(t_{b+|w|},q_{b+|w|}))\in (T\times Q)^{|w|}$ as the leading head reads the suffix $w$ of $yw$. The function $g$ can then be calculated as
    \[g(z)=\max_{w\sqsubseteq z}\prod_{i=1}^{|w|}2\beta(q_{b+i})(w[i]).\]
    Hence, we can define a Turing machine $M$ that, given input
    $(j,t,q,u)\in\N\times T\times Q\times\{0,1\}^{n+1}$,
    finds the lexicographically $j$\textsuperscript{th} string $z\in\{0,1\}^{n-b+1}$ such that $g(z)>2^{\varepsilon\cdot (n-b)}$ and prints $0^{b}z\in\{0,1\}^{n+1}$; the rationals $\gamma$ and $\varepsilon$ are included in the machine description, $n=|u|-1$, and $b=\lfloor\gamma n\rfloor$+1.
    
    Suppose now that~\eqref{eq:nowin} does not hold, i.e., that $g(x)>2^{\varepsilon\cdot (n-b)}$, and let $j$ be its lexicographic position among all such strings. Then $M(j,t,q,u)$ prints $u_h$, so we have
    \begin{align}
        K(u_h\mid u)&\leq K(j,t,q)+O(1)\notag\\
        &\leq K(j)+K(t)+K(q)+O(1)\notag\\
        &\leq K(j)+O(1)\notag\\
        &\leq \log j+O(1),\label{eq:Kofj}
    \end{align}
    as $j\leq 2^n$ trivially. We now bound $j$ more carefully. By the martingale property, for each $\ell\in\N$ we have
    \[\sum_{w\in\{0,1\}^{\ell}}\frac{d_G(yw)}{d_G(y)}=2^{\ell},\]
    so there are fewer than $2^{\ell-\varepsilon\cdot(n-b)}$ strings $w$ of length $\ell$ such that $\frac{d_G(yw)}{d_G(y)}>2^{\varepsilon\cdot(n-b)}$. By the union bound, the number of strings $z\in\{0,1\}^{n-b+1}$ with $g(z)>2^{\varepsilon\cdot(n-b)}$ is therefore less than
    \[\sum_{\ell=0}^{n-b+1}2^{\ell-\varepsilon\cdot(n-b)}<2^{(1-\varepsilon)(n-b)+2}\]
    As $j$ cannot exceed this number, we can combine this with~\eqref{eq:Kofj} to get
    \[K(u_h\mid u)<(1-\varepsilon)(n-b)+O(1),\]
    which is less than $(1-\varepsilon/2)(n-b)$ when $n$ is sufficiently large. By contrapositive, this completes the proof of the lemma.
\end{proof}

\begin{proof}[Proof of Lemma~\ref{lem:dbound}]
    Let $\gamma$ and $n_0$ be as in Lemma~\ref{lem:disjoint}, and choose a constant $N_0\geq n_0$ that is sufficiently large to apply Lemma~\ref{lem:dtok}. Let $y_0= Y[0..N_0-1].$ For each $k \in \Z^+,$ let $N_k = \lfloor \frac{N_0}{\gamma^k} \rfloor$ and $y_k = Y[N_{k-1}..N_k-1]$.
    
    Now let $n\geq N_0$, and let $k$ be such that $N_{k-1}<n\leq N_{k+1}$. Then
    \[Y[0..n-1]=y_0y_1\cdots y_k Y[N_{k}..n-1].\]
    We bound the value of the martingale $d_G$ on this prefix by iterative application of Lemma~\ref{lem:dtok}.
    \begin{align*}
        d_G(Y[0..n-1]) & = d_G(y_0)\frac{d_G(Y[0..N_1-1])}{d_G(y_0)}\cdots \frac{d_G(Y[0..N_k-1])}{d_G(Y[0..N_{k-1}-1])}\frac{d_G(Y[0..n-1])}{d_G(Y[0..N_{k}-1])} \\ 
    & \leq d_G(y_0)2^{\epsilon(|y_1|+ \cdots |y_{k+1}|)} \tag{by Lemma~\ref{lem:dtok}}\\
    & \leq d_G(y_0) 2^{\epsilon N_{k+1}} \\
    & \le d_G(y_0) 2^{\epsilon n/\gamma} \\
    & \le d_G(y_0) 2^{2 \epsilon n}.\tag{$\gamma\in(1/2,1)$}
    \end{align*}
    Hence, 
    \begin{align*}
        d_G^{(1-2\epsilon)}(Y[0..n-1]) & = 2^{(1-2\epsilon -1)n}d_G(Y[0..n-1]) \\
    & \le d_G(y_0),
    \end{align*}
    which is a constant.
\end{proof}

\section{Appendix: Lemmas and Proofs from Section~\ref{sec:instability}}\label{app:instability}

To prove that $\dimFS{(h)}$ is unstable, we first prove several lemmas about the sequences $X=F'_{h+1}(R)$ and $Z=F''_{h+1}(R)$, where $R$ is a Martin-L\"of random sequence. These lemma statements and proofs are very similar to those of Section~\ref{sec:hierarchy}, but we include them here for completeness.

Fix trailing head speeds $0<\sigma_1<\ldots<\sigma_{h-1}<1$ and a buffer constant $c\in\N$.
Let $n\in\N$. For every set $A\subseteq[0..n]$, define the string $X[A]\in\{0,1\}^{n+1}$ by
\[x[A][i]=\begin{cases}X[i]&\text{if }i\in A\\0&\text{otherwise,}\end{cases}\]
and 
$Z[A]\in\{0,1\}^{n+1}$ by
\[z[A][i]=\begin{cases}Z[i]&\text{if }i\in A\\0&\text{otherwise,}\end{cases}\]
for all $i\in\{0,\ldots,n\}$. 

Let $b=\lfloor\gamma n\rfloor+1$, where $\gamma\in(0,1) \cap \Q$ is a parameter whose value will be defined in Lemma~\ref{lem:disjoint}. Define
\[x_h=x[b..n] \text{ and } z_h=z[b..n],\]
\[U=\bigcup_{i=1}^{h-1}[\sigma_ib-c,\sigma_in+c]\]
and
\[x_u=x[U] \text{ and } z_u=z[U]\]
Intuitively, the string $x_u$ contains all bits that might be read by any trailing head while the leading head $h$ reads $x_h$. The same holds for string $z$.

Let $d\in\mathbb{N}$. For each $1\leq i\leq h$, define
\[V_i=\left\{(p_i/p_{h+1})^{\nu_{h+1}(m)}m \;\middle|\; m\in[b,n]\text{ and }\nu_{h+1}(m)\in[0,d]\right\}\]
and its \emph{closure}
\[\overline{V}_i=\left\{m\cdot \frac{p_{h+1}^{a_1+\ldots+a_h}}{p_1^{a_1}\cdots p_h^{a_h}}\;\middle|\; m\in V_i\right\}.\] 

Define $x_i=x[V_i]$ and $\overline{x}_i=x[\overline{V}_i]$, correspondingly, $z_i=z[V_i]$ and $\overline{z}_i=z[\overline{V}_i]$.

\begin{lemma}\label{lem:3inequxz}
Fix $\gamma$ and $j$ as in Lemma~\ref{lem:disjoint}, and assume that $n$ is sufficiently large. Define
\[W=\{0,\ldots,b-1\}\setminus \overline{V}_j.\]
and
\[w_x=x[W] \text{ and } w_z=z[W].\]
If $1\le j \le h-1$, then we have:
    \begin{align}
    K(x_h\mid x_u)&\ge K(x_h\mid w_x) - O(1) \label{geq:1stx} \\
    &\ge  K(x_j\mid w_x) - O(\log n) \label{geq:2ndx} \\
    &\ge (n-b)(1-p_{h+1}^{-d-1}) - O(\log n). \label{geq:3rdx}
\end{align}
Similarly, if $2 \le j \le h$, we have:
\begin{align}
    K(z_h\mid z_u)&\ge K(z_h\mid w_z) - O(1) \label{geq:1stz} \\
    &\ge  K(z_j\mid w_z) - O(\log n) \label{geq:2ndz} \\
    &\ge (n-b)(1-p_{h+1}^{-d-1}) - O(\log n). \label{geq:3rdz}
\end{align}
\end{lemma}

\begin{proof}
    For \eqref{geq:1stx}, Since $U\cap \overline{V}_j=\emptyset$ by using Lemma~\ref{lem:disjoint}, and $W=\{0,\ldots,b-1\}\setminus \overline{V}_j,$ then $U\subseteq W.$

Then by using~\eqref{eq:dpi} and~\eqref{eq:soi}, we have $ K(x_h\mid x_u) + O(1) \ge K(x_h\mid w_x)$.

The same reasoning works for \eqref{geq:1stz}, hence, we have $ K(z_h\mid z_u) + O(1) \ge K(z_h\mid w_z)$.

\bigskip
For \eqref{geq:2ndx} and \eqref{geq:2ndz}, first, we need to show the following claim:
\begin{claim}\label{Claim:comp_app}
    $(x_h,w_x)\mapsto x_j$ is a computable mapping. Correspondingly, $(z_h,w_z)\mapsto z_j$ is also a computable mapping.
\end{claim}
\begin{claimproof}
    Note here, for computing $x_j$, $1\le j \le h-1$.

For each $k\in\mathbb{N}$, let
    \[V^{(k)}_j=\left\{t\in V_j\mid \nu_j(t)\leq k\right\}\]
    and
    $x^{(k)}_j=x\big[V^{(k)}_j\big]$. We show by induction on $k$ that, for all $k\in\mathbb{N}$, the mapping $(x_h,w_x)\mapsto x^{(k)}_j$ is computable. For the basis step, simply observe that $V^{(0)}_j\subseteq[b..n]$.
    
    For the inductive step, fix $k\in\mathbb{N}$, assume $(x_h,w_x)\mapsto x^{(k)}_j$ is computable, and let $r\in V_j^{(k+1)}$. Then $r=s\cdot p_j^e$ for some $e\in\{0,\ldots,k+1\}$ and $s$ satisfying $\nu_{h+1}(s)=0$ and $s\cdot p_{h+1}^e\in[b,n]$. If $e=0$, then $r\in[b,n]$. Otherwise,
    \begin{equation}\label{eq:orgx}
        x[s\cdot p_{h+1}^e]=\bigoplus_{i_1,\ldots,i_e\in\{1,\ldots,h-1\}}x\left[s\cdot p_{i_1}\cdots p_{i_e}\right],
    \end{equation}
    which can be rearranged to obtain a formula for $x[r]=x[s\cdot p_j^e]$:
    \begin{equation}\label{eq:recoverx}
        x[s\cdot p_{h+1}^e]\oplus\bigoplus_{\substack{i_1,\ldots,i_e\in\{1,\ldots,h-1\}\\(i_1,\ldots,i_e)\neq(j,\ldots,j)}}x\left[s\cdot p_{i_1}\cdots p_{i_e}\right].
    \end{equation}
    Therefore, it suffices to show that every index appearing in~\eqref{eq:recoverx} belongs to
    \[[b,n]\cup W\cup V^{(k)}_j.\]
    The index $s\cdot p_{h+1}^e$ of the first term of~\eqref{eq:recoverx} belongs to $[b,n]$. Each of the remaining indices can be written in the form $t=s\cdot p_j^{\ell}p_{i_1}\cdots p_{i_{e-\ell}}$, for some $0\leq\ell\leq e-1$ and $i_1,\ldots,i_{e-\ell}\in\{1,\ldots,h-1\}\setminus\{j\}$. If $t\geq b$, then $t<s\cdot p_{h+1}^e\leq n$ implies $t\in[b,n]$; hence, assume $t<b$. If $t\not\in W$, then $t\in\overline{V}_j$, and it follows that $t\in V_j$ since
    \[\nu_{h+1}(t)=\nu_{h+1}(s\cdot p_j^{\ell}p_{i_1}\cdots p_{i_{e-\ell}})=\nu_{h+1}(s)=0.\]
    In particular, since
    \[\nu_j(t)=\nu_j(s)+\ell\leq \nu_j(s)+e-1\leq \nu_j(r)-1\leq k,\]
    we have $t\in V_j^{(k)}$ in this case. Thus, each bit of $x$ in~\eqref{eq:recoverx} can be obtained from $x_h$, $w_x$, or $x_j^{(k)}$. By induction, we conclude that the mapping $(x_h,w_x)\mapsto x_j$ is computable.

    For showing $(z_h,w_z)\mapsto z_j$ is computable, the idea is identical, we only need to note that $2\le j \le h$ here.

    Equation~\eqref{eq:orgx} needs to be changed to: 
    \[z[s\cdot p_{h+1}^e]=\bigoplus_{i_1,\ldots,i_e\in\{2,\ldots,h\}}z\left[s\cdot p_{i_1}\cdots p_{i_e}\right],\]
    and equation~\eqref{eq:recoverx} needs to be changed to:
    \[z[s\cdot p_{h+1}^e]\oplus\bigoplus_{\substack{i_1,\ldots,i_e\in\{2,\ldots,h\}\\(i_1,\ldots,i_e)\neq(j,\ldots,j)}}z\left[s\cdot p_{i_1}\cdots p_{i_e}\right].\]
    The reasoning is the same.
\end{claimproof}
   
    Therefore,
\begin{align*}
    K(x_j\mid w_x)
    &\leq K(x_j,w_x)-K(w_x)+O(\log n)\tag{symmetry of information}\\
    &\leq K(x_h,w_x)-K(w_x)+O(\log n)\tag{by~\eqref{eq:dpi} and Claim~\ref{Claim:comp_app}}\\
    &\leq K(x_h\mid w_x)+O(\log n).\tag{symmetry of information}
\end{align*}
Similarly, $K(z_j\mid w_z) \leq K(z_h\mid w_z)+O(\log n)$.

\bigskip
For \eqref{geq:3rdx}, first, observe that
\begin{align*}
    |V_j|&=\left|\left\{(p_j/p_{h+1})^{\nu_{h+1}(m)}m \;\middle|\; m\in[b,n]\text{ and }\nu_{h+1}(m)\in[0,d]\right\}\right|\\
    &=\left|\left\{m\in[b,n]\mid\nu_{h+1}(m)\in[0,d]\right\}\right|\\
    &=(n-b+1)-\left|\left\{m\in[b,n]\mid\nu_{h+1}(m)\geq d+1\right\}\right|\\
    &\geq (n-b+1)-\left\lceil\frac{n-b+1}{p_{h+1}^{d+1}}\right\rceil\\
    &\geq (n-b)\left(1-p_{h+1}^{-d-1}\right)-O(1).
\end{align*}
It suffices to show that the bits at these indices are nearly incompressible given $w_x$, in the sense that $K(x_j\mid w_x)\geq |V_j|-O(\log n)$.

Let $m=\left\lfloor (p_{h+1}-1) \frac{n}{p_{h+1}}\right\rfloor+(n\bmod p_{h+1})$. Since each index $t\in V_j$ has $\nu_{h+1}(t)=0$, each nontrivial bit of $x_j$ is copied directly from a bit in $R$. Let $A\subseteq[0,m]$ be the set of indices in $R$ of those bits, so there is a computable bijection between $x_j$ and $R[A]$. Furthermore, the set $A$ is computable given $m$ and satisfies $|A|=|V_j|$. Let $B=[0,m]\setminus A$. Since $W$ and $\overline{V}_j$ are disjoint, there is a computable mapping from $R[B]$ to $w_x$. Furthermore, since $R$ is random, we have
\begin{equation}\label{eq:rarbx}
    K(R[A],R[B])\geq K(R[0..m])-O(1)\geq m-O(1).
\end{equation}
It follows that
\begin{align*}
    K(x_j\mid w_x)&\geq K(x_j\mid R[B])-O(1)\tag{$R[B]\mapsto w_x$ is computable}\\
    &\geq K(R[A]\mid R[B])-O(1)\tag{$x_j\mapsto R[A]$ is computable}\\
    &\geq K(R[A],R[B])-K(R[B])-O(\log n)\tag{symmetry of information}\\
    &= m-K(R[B])-O(\log n)\tag{by~\eqref{eq:rarbx}}\\
    &\geq m-|B|-O(\log n)\tag{$R[B]$ is computable from $|B|$ bits}\\
    &=|A|-O(\log n)\\
    &=|V_j|-O(\log n).
\end{align*}

For \eqref{geq:3rdz}, we use the same reasoning to show that $K(z_j\mid w_z)\geq |V_j|-O(\log n)$. Only need to pay attention that the computable functions are between $z_j$ and $R[A]$, between $R[B]$ and $w_z$.
\end{proof}

\begin{lemma}\label{lem:dtokx}
    Let $h\in\Z^+$, $G$ an $h$-FSG, $\varepsilon>0$ a rational constant, and $n\in\N$ sufficiently large. Define $b$, $x_u$, and $x_h$ as above, let $v=x_h[b..n]$, and let $x=X[0..b-1]$. If
    \[K(x_h\mid x_u)\geq (n-b)(1-\frac{\varepsilon}{2})\]
    then
    \begin{equation}\label{eq:nowinx}
        \max_{w\sqsubseteq v}\frac{d_G(xw)}{d_G(x)}\leq 2^{\varepsilon\cdot(n-b)}.
    \end{equation}
\end{lemma} 

\begin{proof}
    First observe that given $u$ and the state $(t_b,q_b)$ of $G$ after the leading head has read $x$, we can compute the function $g:\{0,1\}^{(n-b+1)}\to\Q$ defined by
    \[g(z)=\max_{w\sqsubseteq z}\frac{d_G(xw)}{d_G(x)},\]
    even without knowing $x$. This is because, for each $w\sqsubseteq z$, as $u$ contains all information the trailing heads could possibly access while the leading head's position is in $[b,n]\supseteq [b,b+|w|]$, we can simulate $G$ to determine its partial state sequence $((t_{b+1},q_{b+1}),\ldots,(t_{b+|w|},q_{b+|w|}))\in (T\times Q)^{|w|}$ as the leading head reads the suffix $w$ of $xw$. The function $g$ can then be calculated as
    \[g(z)=\max_{w\sqsubseteq z}\prod_{i=1}^{|w|}2\beta(q_{b+i})(w[i]).\]
    Hence, we can define a Turing machine $M$ that, given input
    \[(j,t,q,x_u)\in\N\times T\times Q\times\{0,1\}^{n+1},\]
    finds the lexicographically $j$\textsuperscript{th} string $z\in\{0,1\}^{n-b+1}$ such that $g(z)>2^{\varepsilon\cdot (n-b)}$ and prints $0^{b}z\in\{0,1\}^{n+1}$; the rationals $\gamma$ and $\varepsilon$ are included in the machine description, $n=|x_h|-1$, and $b=\lfloor\gamma n\rfloor$+1.
    
    Suppose now that~\eqref{eq:nowinx} does not hold, i.e., that $g(v)>2^{\varepsilon\cdot (n-b)}$, and let $j$ be its lexicographic position among all such strings. Then $M(j,t,q,x_u)$ prints $x_h$, so we have
    \begin{align}
        K(x_h\mid x_u)&\leq K(j,t,q)+O(1)\notag\\
        &\leq K(j)+K(t)+K(q)+O(1)\notag\\
        &\leq K(j)+O(1)\notag\\
        &\leq \log j+O(1),\label{eq:Kofjx}
    \end{align}
    as $j\leq 2^n$ trivially. We now bound $j$ more carefully. By the martingale property, for each $\ell\in\N$ we have
    \[\sum_{w\in\{0,1\}^{\ell}}\frac{d_G(xw)}{d_G(x)}=2^{\ell},\]
    so there are fewer than $2^{\ell-\varepsilon\cdot(n-b)}$ strings $w$ of length $\ell$ such that $\frac{d_G(xw)}{d_G(x)}>2^{\varepsilon\cdot(n-b)}$. By the union bound, the number of strings $z\in\{0,1\}^{n-b+1}$ with $g(z)>2^{\varepsilon\cdot(n-b)}$ is therefore less than
    \[\sum_{\ell=1}^{n-b+1}2^{\ell-\varepsilon\cdot(n-b)}<2^{(1-\varepsilon)(n-b)+2}\]
    As $j$ cannot exceed this number, we can combine this with~\eqref{eq:Kofjx} to get
    \[K(x_h\mid x_u)<(1-\varepsilon)(n-b)+O(1),\]
    which is less than $(1-\varepsilon/2)(n-b)$ when $n$ is sufficiently large. By contrapositive, this completes the proof of the lemma.
\end{proof}

\begin{lemma}\label{lem:dtokz}
    Let $h\in\Z^+$, $G$ an $h$-FSG, $\varepsilon>0$ a rational constant, and $n\in\N$ sufficiently large. Define $b$, $z_u$, and $z_h$ as above, let $v=z_h[b..n]$, and let $z=Z[0..b-1]$. If
    \[K(z_h\mid z_u)\geq (n-b)(1-\frac{\varepsilon}{2})\]
    then
    \begin{equation}\label{eq:nowinz}
        \max_{w\sqsubseteq v}\frac{d_G(zw)}{d_G(z)}\leq 2^{\varepsilon\cdot(n-b)}.
    \end{equation}
\end{lemma} 
\begin{proof}
    The same proof for Lemma~\ref{lem:dtokx} works here.
\end{proof}

\begin{lemma}\label{lem:dboundxz}
    Given $X = F'_{h+1}(R)$ and $Z = F''_{h+1}(R)$, where $R$ is a Martin-L\"of random sequence and $h \in \mathbb{Z}^+$. For every $\epsilon \in \Q \cap (0,1)$ and every $h$-head finite-state gambler $G$, exactly one of the following is true: $\{d_G^{(1-2\epsilon)}(X[1..n])\mid n \in \N \}$ is bounded, $\{d_G^{(1-2\epsilon)}(Z[1..n])\mid n \in \N \}$ is bounded, or both of them are bounded.
\end{lemma}

\begin{proof}
    Let $\gamma$ and $n_0$ be as in Lemma~\ref{lem:disjoint}, and choose a constant $N_0\geq n_0$ that is sufficiently large to apply Lemma~\ref{lem:dtokx}. Let $x_0= X[0..N_0-1].$ For each $k \in \Z^+,$ let $N_k = \lfloor \frac{N_0}{\gamma^k} \rfloor$ and $x_k = X[N_{k-1}..N_k-1]$.
    
    Now let $n\geq N_0$, and let $k$ be such that $N_{k-1}<n\leq N_{k+1}$. Then
    \[X[0..n-1]=x_0x_1\cdots x_k X[N_{k}..n-1].\]
    When $1 \le j \le h-1$, we bound the value of the martingale $d_G$ on this prefix by iterative application of Lemma~\ref{lem:dtokx}.
    \begin{align*}
        d_G([X[0..n-1]) & = d_G(x_0)\frac{d_G(X[0..N_1-1])}{d_G(x_0)}\cdots \frac{d_G(X[0..N_k-1])}{d_G(X[0..N_{k-1}-1])}\frac{d_G(X[0..n-1])}{d_G(X[0..N_{k}-1])} \\ 
    & \leq d_G(x_0)2^{\epsilon(|x_1|+ \cdots |x_{k+1}|)} \tag{by Lemma~\ref{lem:dtokx}}\\
    & \leq d_G(x_0) 2^{\epsilon N_{k+1}} \\
    & \le d_G(x_0) 2^{\epsilon n/\gamma} \\
    & \le d_G(x_0) 2^{2 \epsilon n}.\tag{$\gamma\in(1/2,1)$}
    \end{align*}
    Hence, 
    \begin{align*}
        d_G^{(1-2\epsilon)}([X[0..n-1]) & = 2^{(1-2\epsilon -1)n}d_G([X[0..n-1]) \\
    & \le d_G(x_0),
    \end{align*}
    which is a constant.

    Use the same definition for $Z$, \[Z[0..n-1]=z_0z_1\cdots z_k Z[N_{k}..n-1].\] We can bound the value of the martingale $d_G$ on this prefix by iterative application of Lemma~\ref{lem:dtokz}. By the same reasoning, we have $d_G^{(1-2\epsilon)}([Z[0..n-1]) \le d_G(z_0),$ which is a constant.    
\end{proof}

\begin{lemma}\label{lem:dimxz}
   Given $X = F'_{h+1}(R)$ and $Z = F''_{h+1}(R)$, where $R$ is a Martin-L\"of random sequence and $h \in \mathbb{Z}^+$, \[\dimFS{(h)}(X) \le \DimFS{(h)}(X) \le 1-\frac{1}{p_{h+1}}\] and \[\dimFS{(h)}(Z) \le \DimFS{(h)}(Z) \le 1-\frac{1}{p_{h+1}}.\]
\end{lemma}
\begin{proof}
    Let $G_x, G_z$ be $h$-head finite-state gamblers and let $d_{G_x}, d_{G_z}$ be their corresponding martingales. 
    
    The idea of constructing $G_x$ and $G_z$ is similar to the idea of constructing $G$ for proving Lemma~\ref{lem:dim}. 
    
    Define $G_x = (T \times Q, \Sigma, \delta, \mu, \beta, (t_0, q_0), c_0)$ as below:
    \begin{itemize}
        \item $\Sigma = \{0,1\}.$
        \item $|T\times Q| = 2p_{h+1}-1, T=\{t_0, t_1, t_2,\cdots, t_{p_{h+1}}\}, Q = \{q_0, q_1, \cdots, q_{2p_{h+1}-2}\}$.
        \item $\mu(t_0) = \mu(t_1) = \mu(t_{p_{h+1}}) = \{\overbrace{1\;\cdots\;1}^{h-1}\}$;
        
        $\forall 2 \le i \le h, \; \forall p_{i-1}-1<j\le p_i-1, \mu(t_j)=\{\overbrace{0\;\cdots\;0}^{i-1}\;1\;\cdots\;1\}$.
        
        When $i=h+1, \; \forall p_{i-1}-1<j\le p_i-1, \mu(t_j)=\{\overbrace{0\;\cdots\;0}^{h-1}\}$.
        \item $\forall 0 \le i \le p_{h+1}-1, \delta_T(t_i) = t_{i+1}$ and $\delta_T(t_{p_{h+1}}) = t_0.$
        \item For all $\vec{\sigma} \in \Sigma^h, \; \delta_Q(q_0, \vec{\sigma}) = q_1, \delta_Q(q_1, \vec{\sigma}) = q_2$; 

        Let $\vec{\sigma}[i]$ be the $i$\textsuperscript{th} element in the vector $\vec{\sigma}$,
        \[
    \delta_Q(q_2, \vec{\sigma}) = 
    \begin{cases}
    q_3, & \text{if } \vec{\sigma}[1]=1,\\
    q_4,  & \text{if } \vec{\sigma}[1]=0.
    \end{cases}
    \]
    $\forall 2 \le i \le h-1,\text{ Let } j= 2(p_i-1)$, then
    \[
    \delta_Q(q_j, \vec{\sigma}) = 
    \begin{cases}
    q_{j+2}, & \text{if } \vec{\sigma}[i]=0,\\
    q_{j+1},  & \text{if } \vec{\sigma}[i]=1.
    \end{cases}
    \]
    and 
    \[
    \delta_Q(q_{j-1}, \vec{\sigma}) = 
    \begin{cases}
    q_{j+2}, & \text{if } \vec{\sigma}[i]=1,\\
    q_{j+1},  & \text{if } \vec{\sigma}[i]=0.
    \end{cases}
    \]
    $\forall \vec{\sigma} \in \Sigma^h, \;  \delta_Q(q_{2p_{h+1}-2}, \vec{\sigma}) = \delta_Q(q_{2p_{h+1}-3}, \vec{\sigma})=q_1$.
    
    For all other states, $\forall \vec{\sigma} \in \Sigma^h, \;  \delta_Q(q_i, \vec{\sigma}) = q_{i+2}$.
    \item $\beta(q_{2p_{h+1}-2})(0)=1, \beta(q_{2p_{h+1}-2})(1)=0, \beta(q_{2p_{h+1}-3})(0)=0, \beta(q_{2p_{h+1}-3})(1)=1$; $\forall q \in Q\setminus\{q_{2p_{h+1}-2}, q_{2p_{h+1}-3}\}, \beta(q)(0) = \beta(q)(1) = \frac{1}{2}$.
   \end{itemize}
    
    Based on our construction, \[d_{G_x}(X[1..n]) = 2^{\lfloor n/p_{h+1} \rfloor}\]
    Then we can show $\forall\epsilon \in \Q \cap (0,1), \; \exists \text{ a }(1-\frac{1}{p_{h+1}} + \epsilon)$-gale that succeeds on $X$.

    Since 
    \begin{align}
        d_{G_x}^{\Big(1-\frac{1}{p_{h+1}} + \epsilon\Big)}(X[1..n]) &= 2^{\Big(1-\frac{1}{p_{h+1}} + \epsilon -1\Big)n}d_{G_x}(X[1..n]) \notag\\
        & = 2^{\Big(\epsilon -\frac{1}{p_{h+1}}\Big)n}\times 2^{\lfloor \frac{n}{p_{h+1}} \rfloor} \notag\\
        & \ge 2^{\Big(\epsilon -\frac{1}{p_{h+1}}\Big)n}\times 2^{\frac{n}{p_{h+1}} -1} \notag\\
        & = 2^{\epsilon n -1}.\label{eq:dgw}
    \end{align}
    Therefore, \[\dimFS{(h)}(X) \le \DimFS{(h)}(X) \le 1-\frac{1}{p_{h+1}}.\]
        
    Define $G_z = (T \times Q, \Sigma, \delta, \mu, \beta, (t_0, q_0), c_0)$ as below:
    \begin{itemize}
        \item $\Sigma = \{0,1\}.$
        \item $|T\times Q| = 2p_{h+1}-2, T=\{t_0, t_1, t_2,\cdots, t_{p_{h+1}}\}, Q = \{q_0, q_1, \cdots, q_{2p_{h+1}-3}\}$.
        \item $\mu(t_0) = \mu(t_1) = \mu(t_2)=\mu(t_{p_{h+1}}) = \{\overbrace{1\;\cdots\;1}^{h-1}\}$;
        
        $\forall 3 \le i \le h+1, \; \forall p_{i-1}-1<j\le p_i-1, \mu(t_j)=\{\overbrace{0\;\cdots\;0}^{i-2}\;1\;\cdots\;1\}$.
        \item $\forall 0 \le i \le p_{h+1}-1, \delta_T(t_i) = t_{i+1}$ and $\delta_T(t_{p_{h+1}}) = t_0.$
        \item $\forall \vec{\sigma} \in \Sigma^h, \;  \delta_Q(q_0, \vec{\sigma}) = q_1, \delta_Q(q_1, \vec{\sigma}) = q_2,\delta_Q(q_2, \vec{\sigma}) = q_3$; 
        
         Let $\vec{\sigma}[i]$ be the $i$\textsuperscript{th} element in the vector $\vec{\sigma}$,  
        \[
    \delta_Q(q_3, \vec{\sigma}) = 
    \begin{cases}
    q_4, & \text{if } \vec{\sigma}[1]=0,\\
    q_5,  & \text{if } \vec{\sigma}[1]=1.
    \end{cases}
    \]
    $\forall 2 \le i \le h-1,\text{ Let } j= 2(p_{i+1}-2)$, then
    \[
    \delta_Q(q_j, \vec{\sigma}) = 
    \begin{cases}
    q_{j+2}, & \text{if } \vec{\sigma}[i]=0,\\
    q_{j+3},  & \text{if } \vec{\sigma}[i]=1.
    \end{cases}
    \]
    and 
    \[
    \delta_Q(q_{j+1}, \vec{\sigma}) = 
    \begin{cases}
    q_{j+2}, & \text{if } \vec{\sigma}[i]=1,\\
    q_{j+3},  & \text{if } \vec{\sigma}[i]=0.
    \end{cases}
    \]
    $\forall \vec{\sigma} \in \Sigma^h, \;  \delta_Q(q_{2p_{h+1}-3}, \vec{\sigma}) = \delta_Q(q_{2p_{h+1}-4}, \vec{\sigma})=q_1$.
    
    For all other states, $\forall \vec{\sigma} \in \Sigma^h, \;  \delta_Q(q_i, \vec{\sigma}) = q_{i+2}$.
    \item $\beta(q_{2p_{h+1}-3})(0)=0, \beta(q_{2p_{h+1}-3})(1)=1, \beta(q_{2p_{h+1}-4})(0)=1, \beta(q_{2p_{h+1}-4})(1)=0$; $\forall q \in Q\setminus\{q_{2p_{h+1}-3}, q_{2p_{h+1}-4}\}, \beta(q)(0) = \beta(q)(1)=\frac{1}{2}$.
   \end{itemize}
    
    Based on our construction, \[d_{G_z}(Z[1..n]) = 2^{\lfloor n/p_{h+1} \rfloor}\]
    Then we can show $\forall\epsilon \in \Q \cap (0,1), \; \exists \text{ a }(1-\frac{1}{p_{h+1}} + \epsilon)$-gale that succeeds on $Z$. The computation is identical to~\eqref{eq:dgw}.

    Therefore,
    \[\dimFS{(h)}(Z) \le \DimFS{(h)}(Z) \le 1-\frac{1}{p_{h+1}}.\]
\end{proof}

\begin{proof}[Proof of Theorem~\ref{thm:unstable}]
Fix any $h\in\N$ and any Martin-L\"of random sequence $R\in\C$, and let $X=F'_{h+1}(R)$ and $Z=F''_{h+1}(R)$.

Based on Lemma~\ref{lem:dimxz}, we can find $h$-head finite-state gambler $G_X$ and $G_Z$ such that their martingales can succeed on $X$ and $Z$ correspondingly.

Fix $\gamma \text{ and } j$ by Lemma~\ref{lem:disjoint}, such that $U\cap \overline{V}_j=\emptyset$ for some $0\leq j\le h$.

We need to talk about three cases:
\begin{itemize}
    \item $j = 1$;
    \item $2 \le j \le h-1$;
    \item $j = h$.
\end{itemize}

By using Lemma~\ref{lem:dboundxz}, $\forall \epsilon \in \Q \cap (0,1), \forall h$-head finite-state gambler $G$, exactly one of the following is true: $\{d_G^{(1-2\epsilon)}(X[1..n])\mid n \in \N \}$ is bounded, $\{d_G^{(1-2\epsilon)}(Z[1..n])\mid n \in \N \}$ is bounded, both of them are bounded. Letting $\epsilon$ approach 0, we have
\[\dimFS{(h)}(\{X, Z\})=\DimFS{(h)}(\{X,Z\})=1,\]
Consequently,
\[\dimFS{(h)}(\{X, Z\})>\max\{\dimFS{(h)}(\{X\}), \dimFS{(h)}(\{Z\})\}\]
and
\[\DimFS{(h)}(\{X, Z\})>\max\{\DimFS{(h)}(\{X\}), \DimFS{(h)}(\{Z\})\}.\]
\end{proof}

\section{Appendix: Proof of Theorem~\ref{thm:appd}}\label{app:stability}

\begin{proof}[Proof of Theorem~\ref{thm:appd}]
    For $j \in \{1,2\}$, set $G_j = (T_j \times Q_j, \Sigma,\delta_j, \mu_j, \beta_j, (t_j, q_j), c_j)$, also the corresponding transition functions $\delta_{T_j}, \delta_{Q_j}$ and $d_j = d_{G_j}$. Assume without loss of generality that $T_1 \cap T_2 = \emptyset$, $Q_1 \cap Q_2 = \emptyset$, and $c_1=c_2=1$.

    Let $h = h_1+h_2-1$ and let $\epsilon$ be as given, let $r = \lceil 1- \log_2 (1 - 2^{-\frac{\epsilon}{2}})\rceil$,
    and let $\mathbb{Q}_r$ be the set of all $r$-dyadic rational numbers, i.e., all rational numbers $\zeta$ such that $2^r \zeta \in \mathbb{Z}$.

    Define the rounding operator $[\cdot]_r: [0,1] \rightarrow \mathbb{Q}_r$ by 
    \[
    [x]_r = 
    \begin{cases}
    2^{-r} \lceil 2^rx \rceil, & \text{if } x \le \frac{1}{2},\\
    2^{-r} \lfloor 2^rx \rfloor,  & \text{if } x > \frac{1}{2}.
    \end{cases}
    \]

    Note that given the exact capital allocation ratio $\frac{d_1(w)}{d_1(w)+d_2(w)}$, we could compute the betting function that implements the exact average $\frac{d_1(w) + d_2(w)}{2}$. Unfortunately, this ratio cannot, in general, be stored within finite states. Hence, we maintain an approximate allocation ratio $\alpha$ as part of the gambler’s state.

    Define an $h$-FSG $G = (T \times Q, \Sigma,\delta, \mu, \beta, ((t_1,t_2), (q_1,q_2, \frac{1}{2})), c_0)$, whose components are defined as follows:
    \begin{itemize}
        \item $T = T_1 \times T_2$, $Q = Q_1 \times Q_2 \times (\mathbb{Q}_r \cap [0,1])$.
        \item $\delta: T \times Q \times \Sigma^h \rightarrow T \times Q$ is defined, for all $(t', q') \in T_1 \times Q_1, (t'', q'') \in T_2 \times Q_2$, and $\alpha \in \mathbb{Q}_r \cap [0,1]$, by
        \begin{align*}
            \delta_T(t', t'') &= (\delta_{T_1}(t'), \delta_{T_2}(t'')),\\
            \delta_Q((q',q'',\alpha), \vec\sigma_1[1..h_1-1]\cdot \vec\sigma_2)&= \big(\delta_{Q_1}(q',\vec\sigma_1), \delta_{Q_2}(q'', \vec\sigma_2), \alpha' \big),
           \end{align*}
           where $\vec\sigma_1[1..h_1-1]\cdot \vec\sigma_2$ is the concatenation of those two vectors and
           \[\alpha'=\Bigg[\frac{\alpha \beta_1(q')(\vec\sigma_2[h_2])}{\alpha \beta_1(q')(\vec\sigma_2[h_2]) + (1-\alpha)\beta_2(q'')(\vec\sigma_2[h_2])}\Bigg]_r.\]
           Note that $\vec\sigma_1[h_1] = \vec\sigma_2[h_2]$ because both gamblers' leading heads are in the same position.

        \item $\mu: T_1 \times T_2 \rightarrow \{0,1\}^{h_1 + h_2 -2}$ is defined, for all $t' \in T_1, t'' \in T_2$, by
        \[\mu(t', t'') = \mu_1(t') \cdot \mu_2(t'').\]
        \item $\beta:Q\to\Delta_\Q(\Sigma)$ is defined, for all $q' \in Q_1$, $q'' \in Q_2$, $\alpha \in \mathbb{Q}_r \cap [0,1]$, and $b\in \Sigma$, by
        \[\beta(q',q'',\alpha)(b)= \alpha \beta_1(q')(b) + (1-\alpha)\beta_2(q'')(b).\]
        \item The initial state is $((t_1,t_2), (q_1,q_2, \frac{1}{2})) \in T \times Q$.
        \item The initial capital is $c_0=\frac{c_1+c_2}{2}=1$.
    \end{itemize} 
    Let $d = d_{G}$ and $k = |\Sigma|$.
$d(w)$ is recursively defined by $d(\lambda)= c_0$ and, for all $w \in \Sigma^*$ and $b \in \Sigma$ if $G$ reaches the state $(q',q'', \alpha)$ after reading $w$, then
\[d(wb) = k d(w)\beta(q',q'',\alpha)(b).\]

Define $\tilde{d}_1: \Sigma^* \rightarrow \mathbb{Q} \cap [0,\infty)$ and $\tilde{d}_2: \Sigma^* \rightarrow \mathbb{Q} \cap [0,\infty)$ as follows:
\[\tilde{d}_1(wb) = 2d(wb) \times \frac{\alpha \beta_1(q')(b)}{\alpha \beta_1(q')(b) + (1-\alpha)\beta_2(q'')(b)},\]
and 
\begin{align*}
    \tilde{d}_2(wb)
    &= 2d(wb) \times \frac{(1-\alpha) \beta_2(q'')(b)}{\alpha \beta_1(q')(b) + (1-\alpha)\beta_2(q'')(b)},
\end{align*}
where $\tilde{d}_1(\lambda)= \tilde{d}_2(\lambda) = 1$.

We now analyze the relationship between $d$, $d_1$, and $d_2$. Let
\[\hat{\alpha}= \frac{\alpha \beta_1(q')(b)}{\alpha \beta_1(q')(b) + (1-\alpha)\beta_2(q'')(b)}.\]
Each $\alpha$ has a corresponding $\hat{\alpha}$ which is in the same iteration, and $\alpha = [\hat{\alpha}]_r$. In the initial state, when the value of $\alpha$ is $1/2$, we set the corresponding $\hat{\alpha}$ to be $1/2$ as well. Hence, we can also write the formulas for $\tilde{d}_1$ and $\tilde{d}_2$ as
\[\tilde{d}_1(wb) = 2d(wb) \times \hat{\alpha}\quad\text{and}\quad\tilde{d}_2(wb) = 2d(wb) \times (1-\hat{\alpha}).\]

Now we show that for all $w \in \Sigma^*$ and $b \in \Sigma$,
\[\tilde{d}_1(wb) = k \tilde{d}_1(w) \frac{\alpha}{\hat{\alpha}} \beta_1(q')(b)\quad\text{and}\quad\tilde{d}_2(wb) = k \tilde{d}_2(w) \frac{1-\alpha}{1-\hat{\alpha}} \beta_2(q'')(b).\]
For the first,
\begin{align*}
    \tilde{d}_1(wb) &= 2d(wb) \times \frac{\alpha \beta_1(q')(b)}{\alpha \beta_1(q')(b) + (1-\alpha)\beta_2(q'')(b)} \tag{def. of $\tilde{d_1}$}\\
    &= 2 k d(w)\beta(q',q'',\alpha)(b) \times \frac{\alpha \beta_1(q')(b)}{\alpha \beta_1(q')(b) + (1-\alpha)\beta_2(q'')(b)}\tag{def. of  $d$}\\
    &= 2 k d(w) \alpha \beta_1(q')(b)\tag{def. of $\beta$}\\
    &= 2k \frac{\tilde{d_1}(w)}{2\hat{\alpha}} \alpha \beta_1(q')(b) \tag{def. of $\tilde{d_1}$}\\
    &= k \tilde{d}_1(w) \frac{\alpha}{\hat{\alpha}} \beta_1(q')(b),
\end{align*}
and similarly for the second.

Based on the definition of $[\cdot]_r$, we know that $\hat{\alpha}$ is rounded down to $\alpha$ only when $\hat{\alpha}> \frac{1}{2}$. Hence, $\alpha > \hat{\alpha} - 2^{-r}$. Therefore, both $\frac{\alpha}{\hat{\alpha}}$ and $\frac{1-\alpha}{1-\hat{\alpha}}$ are strictly greater than $1-2^{1-r}$.

Next, we want to show the following:
    \begin{align}
    \tilde{d}_1(w) & \ge (1-2^{1-r})^{|w|}d_1(w) \label{eq:dd1},\\
    \tilde{d}_2(w) & \ge (1-2^{1-r})^{|w|}d_2(w) \label{eq:dd2}.
    \end{align}
We use induction to show that \eqref{eq:dd1} holds for all $w$; the argument for \eqref{eq:dd2} is identical. The base case is that \eqref{eq:dd1} holds if $w = \lambda$, since $\tilde{d}_1(w) = 1$ and
\[(1-2^{1-r})^{|w|}d_1(w) = (1-2^{1-r})^0 \cdot 1 = 1.\]
For the inductive step, let $\ell\in\N$ and assume \eqref{eq:dd1} holds for all $w \in \{0,1\}^\ell$. Let $w'= wb$, where $w \in \Sigma^\ell$ and $b \in \Sigma$. Then we have the following:
\begin{align*}
    \tilde{d}_1(wb) &= k \tilde{d}_1(w) \frac{\alpha}{\hat{\alpha}} \beta_1(q')(b) \\
    & \ge k\times (1-2^{1-r})^{|w|}d_1(w) \frac{\alpha}{\hat{\alpha}} \beta_1(q')(b) \\
    &= kd_1(w)\beta_1(q')(b)\times (1-2^{1-r})^{|w|}\frac{\alpha}{\hat{\alpha}} \\
    &= d_1(wb) \times (1-2^{1-r})^{|w|}\frac{\alpha}{\hat{\alpha}} \\
    & \geq  d_1(wb) \times (1-2^{1-r})^{|w|} (1-2^{1-r})\\
    & = d_1(wb) \times (1-2^{1-r})^{|w|+1} \\
    & = d_1(wb) \times (1-2^{1-r})^{|wb|}.
 \end{align*}

Therefore, when we choose $w$ such that its length is sufficiently large, satisfying $\frac{1}{|w|} \le \frac{\epsilon}{2}$ and choose $r =\lceil 1- \log_2 (1 - 2^{-\frac{\epsilon}{2}})\rceil$, then we have:
\begin{align*}
    d(w) &= \frac{\tilde{d}_1(w) + \tilde{d}_2(w)}{2} \\
    & \ge \frac{(1-2^{1-r})^{|w|}}{2}(d_1(w) + d_2(w)) \\
    & \ge 2^{-\epsilon |w|}(d_1(w) + d_2(w)).
\end{align*}

It follows that
$2^{(s+\epsilon -1)|w|}d(w) \ge 2^{(s-1)|w|} d_1(w)$
and $2^{(s+\epsilon -1)|w|}d(w) \ge 2^{(s-1)|w|} d_2(w)$ hold for all sufficiently long strings $w$. Hence,
\[S^{\infty}\big[d_{G_1}^{(s)}\big] \cup S^{\infty}\big[d_{G_2}^{(s)}\big] \subseteq S^{\infty}\big[d_{G}^{(s+\epsilon)}\big]\quad\text{and}\quad S_{\strong}^{\infty}\big[d_{G_1}^{(s)}\big] \cup S_{\strong}^{\infty}\big[d_{G_2}^{(s)}\big] \subseteq S_{\strong}^{\infty}\big[d_{G}^{(s+\epsilon)}\big].\]
\end{proof}

\end{document}